\documentclass[10pt,journal]{IEEEtran}

\makeatletter
\def\ifundefined{\@ifundefined}
\makeatother

\newcommand{\chapter}{}

\usepackage{epsfig}
\usepackage{amsmath}
\usepackage{amssymb} 
\usepackage{graphicx}
\usepackage{amsthm}
\usepackage{mathrsfs}
\usepackage{bbold}
\usepackage{multirow}
\usepackage{rotating}
\usepackage{cite}
\usepackage{ifthen}

\DeclareMathOperator{\sgn}{sgn}
\DeclareMathOperator{\const}{const}

\newboolean{dcol}
\setboolean{dcol}{true}
\newcommand{\bitm}{\begin{itemize}}
\newcommand{\eitm}{\end{itemize}}
\newcommand{\be}{\begin{equation}}
\newcommand{\ee}{\end{equation}}
\newcommand{\bea}{\begin{eqnarray}}
\newcommand{\eea}{\end{eqnarray}}
\newcommand\ba{\renewcommand{\arraystretch}{.7} \left[ \begin{array}{*{6}{@{\hspace{2pt}}r@{\hspace{1.5pt}}}}}
\def\ea{\end{array}\right]}
\def\nn{\nonumber\\}
\newcommand{\bfi}{\begin{figure}}
\newcommand{\efi}{\end{figure}}

\newcommand{\mat}[1]{\begin{bf} #1 \end{bf}}
\newcommand{\fn}{}
\newcommand{\I}[1]{\mathbb{1}\left\{#1 \right\}}
\newcommand{\Norm}[2]{|| #1 ||_{{#2}}}
\newcommand{\Ind}[1]{\mathbb{1} \left\{ #1\right\}}
\newcommand{\U}{U}

\newtheorem{thm}{Theorem}  

\newtheorem{pro}{Proposition}

\newenvironment{thmN}[2][Theorem]{\begin{trivlist} 
\item[\hskip \labelsep {\bfseries #1}\hskip \labelsep {\bfseries #2}] \it }{\end{trivlist}}

\newtheorem{defn}{Definition}       
\newtheorem{rem}{Remark}       

\ifthenelse{\boolean{dcol}}{

\addtolength{\abovedisplayskip}{-1pt}
\addtolength{\belowdisplayskip}{-1pt}

}

\begin{document}

\title{On U-Statistics and Compressed Sensing I: Non-Asymptotic Average-Case Analysis}

\author{
\authorblockN{Fabian Lim${}^*$ and Vladimir Marko Stojanovic}\\
\ifthenelse{\boolean{dcol}}{
\thanks{F. Lim and V. M. Stojanovic are with the Research Laboratory of Electronics, 
Massachusetts Institute of Technology, 77 Massachusetts Avenue, Cambridge, MA 02139.
Email: \{flim,vlada\}@mit.edu.
This work was supported by NSF grant ECCS-1128226. \newline
\indent Part of this work will be presented at the 2012 IEEE International Conference on Communications (ICC), Ottawa, Canada.} 
}{
\authorblockA{Research Laboratory of Electronics\\ Massachusetts Institute of Technology,
77 Massachusetts Avenue, Cambridge, MA 02139\\ 
\{flim,vlada\}@mit.edu 
\thanks{This work was supported by NSF grant ECCS-1128226. 
Part of this work will be presented at the 2012 IEEE International Conference on Communications (ICC), Ottawa, Canada.} 
}}
}
\maketitle

\begin{abstract} 

Hoeffding's U-statistics model combinatorial-type matrix parameters (appearing in CS theory) in a natural way. 
This paper proposes using these statistics for analyzing random compressed sensing matrices, in the non-asymptotic regime (relevant to practice).
The aim is to address certain pessimisms of ``worst-case'' \emph{restricted isometry} analyses, as observed by both Blanchard \& Dossal,~et.~al.

We show how U-statistics can obtain ``average-case'' analyses, by relating to \emph{statistical restricted isometry property} (StRIP) type recovery guarantees.
However unlike standard StRIP, random signal models are not required; the analysis here holds in the \emph{almost sure} (probabilistic) sense.
For Gaussian/bounded entry matrices, 
we show that both $\ell_1$-minimization and LASSO essentially require on the order of $k \cdot [\log((n-k)/u) + \sqrt{2(k/n) \log(n/k)}]$ measurements
to respectively recover at least $1-5u$ fraction, and $1-4u$ fraction, of the signals. 
Noisy conditions are considered.
Empirical evidence suggests our analysis to compare well to Donoho \& Tanner's recent large deviation bounds for $\ell_0/\ell_1$-equivalence, in the regime of block lengths $1000 \sim 3000$ with high undersampling ($50\sim 150$ measurements); similar system sizes are found in recent CS implementation.

In this work, it is assumed throughout that matrix columns are independently sampled. 


\end{abstract}



\ifthenelse{\boolean{dcol}}{
\markboth{Lim and Stojanovic: On U-Statistics and Compressed Sensing I: Non-Asymptotic Average-Case Analysis}{}
}{
\markboth{IEEE Transactions on Signal Processing, Lim and Stojanovic, Prepared using \LaTeX}{}
}


\section{Introduction}


Compressed sensing (CS) analysis involves relatively recent results from random matrix theory~\cite{Candes}, whereby recovery guarantees are framed in the context of matrix parameters known as \emph{restricted isometry constants}. 
Other matrix parameters are also often studied in CS. 
Earlier work on sparse approximation considered a matrix parameter known as \emph{mutual coherence}~\cite{Grib,Elad,Tropp}. 
Fuchs' work on \emph{Karush-Kuhn-Tucker (KKT)} conditions for sparsity pattern recovery considered a parameter involving a matrix \emph{pseudoinverse}~\cite{Fuchs}, re-occurring in recent work~\cite{Tropp,Cand2008,Barga}.
Finally, the \emph{null-space property}~\cite{Zhang,Cohen2008,DAspremont2009} is gaining recent popularity - being the parameter closest related to the fundamental compression limit dictated by \emph{Gel'fand widths}.
All above parameters share a similar feature, that is they are defined over subsets of a certain fixed size $k$.
This combinatorial nature makes them difficult to evaluate, even for moderate block lengths $n$. Most CS work therefore involve some form of randomization to help the analysis. 

While the celebrated $k \log (n/k)$ result was initially approached via asymptotics,
\textit{e.g.},~\cite{Candes,Donoho,RIPsimp,Blanchard}, implementations require finite block sizes.
Hence, non-asymptotic analyses are more application relevant.
In the same practical aspect, recent work deals with non-asymptotic analysis of \emph{deterministic} CS matrices, see~\cite{Calder,Gan,Tropp,Barga}.
On the other hand certain situations may not allow control over the sampling process, whereby the sampling may be inherently random, \textit{e.g.}, prediction of clinical outcomes of various tumors based on gene expressions~\cite{Cand2008}.
Random sampling has certain desirable simplicity/efficiency features - 
see~\cite{Sartipi2011} on data acquisition in the distributed sensor setting.
Also recent hardware implementations point out energy/complexity-cost benefits of implementing \emph{pseudo-random binary sequences}~\cite{Fred,Mishali,EPFL}; these sequences mimic statistical behavior. 
Non-asymptotic analysis is particularly valuable, when random samples are costly to acquire. For example, each clinical trial could be expensive to conduct an excessive number of times. 
In the systems setting, the application could be running on a tight energy budget - whereby processing/communication costs depend on the number of samples acquired.

This work is inspired by the \emph{statistical} notion of the restricted isometry property (StRIP), initially developed for deterministic CS analysis~\cite{Calder,Gan}. The idea is to relax the analysis, by allowing sampling matrix parameters (that guarantee signal recovery) to be satisfied for a \emph{fraction} of subsets. 
Our interest is in ``average-case'' notions in the context of randomized sampling, reason being that certain pessimisms of 
``worst-case'' restricted isometry analyses have been observed in past works~\cite{Blanchard,Dossal2009,Bah}. 
On the other hand in~\cite{DT}, Donoho \& Tanner remarked on potential benefits of the above ``average-case'' notion, recently pursued in an adaptation of a previous asymptotic result~\cite{DTExp}. 
In the multichannel setting, ``average-case'' notions are employed to make analysis more tractable~\cite{F2007,Eldar2009}. 
In~\cite{Epfl2008}~a simple ``thresholding'' algorithm is analyzed via an ``average'' coherence parameter.
However the works in this respect are few, most random analyses are of the ``worst-case'' type, see~\cite{Rudel,RIPsimp,Blanchard,Bah}. 
We investigate the unexplored, with the aim of providing new insights and obtaining new/improved results for the ``average-case''.

Here we consider a random analysis tool that is well-suited to the CS context, yet seemingly left untouched in the literature.
Our approach differs from that of deterministic matrices, where ``average-case'' analysis is typically made accessible via mutual coherence, see~\cite{Calder,Gan,Mishali}. 
For random matrices, we propose an alternative approach via \emph{U-statistics}, which do not require random signal models typically introduced in StRIP analysis, see~\cite{Calder,Epfl2008,Eldar2009}; here, the results are stated in the \emph{almost sure} sense.
U-statistics apply naturally to various kinds of non-asymptotic CS analyses, since they are designed for combinatorial-type parameters.
Also, they have a natural ``average-case''  interpretation, which we apply to recent recovery guarantees 
that share the same ``average-case'' characteristic. 
Finally thanks to the wealth of U-statistical literature, the theory developed here is open to other extensions, \textit{e.g.}, in related work~\cite{Lim2} we demonstrate how U-statistics may also perform ``worst-case'' analysis.


\textbf{Contributions}:  
``Average-case'' analyses are developed based on U-statistics, which are i) empirically observed to have good potential for predicting CS recovery in non-asymptotic regimes, and ii) theoretically obtain measurement rates that incorporate a non-zero failure rate (similar to the $k\log (n/k)$ rate from ``worst-case'' analyses).
We utilize a
U-statistical large deviation concentration theorem, 
under the assumption that the matrix columns are independently sampled.
The large deviation error bound holds \emph{almost surely} (Theorem \ref{thm:Ustat}). No random signal model is needed, and the error is of the order $(n/k)^{-1} \log (n/k)$, whereby $k$ is the U-statistic kernel size (and $k$ also equals sparsity level).
Gaussian/bounded entry matrices are considered.
For concreteness, we connect with StRIP-type guarantees (from~\cite{Barga,Cand2008}) to study the fraction of recoverable signals (\textit{i.e.}, ``average-case'' recovery) of: i) \emph{$\ell_1$-minimization} and ii) \emph{least absolute shrinkage and selection operator (LASSO)}, under noisy conditions.
For both these algorithms we show $\const\cdot k  [\log((n-k)/u) + \sqrt{2(k/n) \log(n/k)}]$ 
measurements are essentially required, to respectively recover at least $1-5u$ fraction (Theorem \ref{thm:L1}), and $1-4u$ fraction (Theorem \ref{thm:Lasso}), of possible signals.
This is improved to $1-3u$ fraction for the noiseless case.
Here $\const =  \max(4/(a_1 a_2)^2, 2c_1/(0.29-a_1)^2)$ for to be specified constants $a_1,a_2,c_1$, where $c_1$ depends on the distribution of matrix entries.
Note that the term $\sqrt{2(k/n) \log(n/k)}$ is at most 1 and vanishes with small $k/n$.
Empirical evidence suggests that our approach compares well with recent results from Donoho \& Tanner~\cite{DTExp} -  improvement is suggested for system sizes found in implementations~\cite{Fred}, with large undersampling (\textit{i.e.}, $m = 50 \sim 100$ and $n = 1000 \sim 3000$).
The large deviation analysis here does show some pessimism in the size of $\const$ above, whereby $\const \geq 4$ (we conjecture possible improvement).
For Gaussian/Bernoulli matrices, we find $\const \approx 1.8$ to be inherently smaller, \textit{e.g.}, for $k=4$ this predicts recovery of $1\times 10^{-6}$ fraction with $153$ measurements - empirically $m= 150$.


\textbf{Note}:  
StRIP-type guarantees~\cite{Barga,Cand2008} seem to work well, by simply \emph{not} placing restrictive conditions on the maximum eigenvalues of the size-$k$ submatrices. 
Our theory applies fairly well for various considered system sizes $k,m,n$ (\textit{e.g.}, Figure \ref{fig:CompL1}), however in \emph{noisy} situations, a (relatively small) factor of $\sqrt{k}$ losses is seen without making certain maximum eigenvalue assumptions. 
For $\ell_1$-recovery, the estimation error is now bounded by a $\sqrt{k}$ factor of its best $k$-term approximation error (both errors measured using the $\ell_1$-norm). 
For LASSO, the 
the non-zero signal magnitudes must now be bounded below by a factor $\sqrt{2 k \log n}$ (with respect to noise standard deviation), as opposed to $\sqrt{2 \log n}$ in~\cite{Cand2008}.
These losses occur not because of StRIP analyses, but because of the estimation techniques employed here.

\textbf{Organization}:  
We begin with relevant background on CS in Section \ref{sect:params}. In Section \ref{sect:Dist} we present a general U-statistical theorem for large-deviation (``average-case'') behavior.
In Section \ref{sect:KKT} the U-statistical machinery is applied to StRIP-type ``average-case'' recovery.
We conclude in Section \ref{sect:conc}.



\newcommand{\Real}{\mathbb{R}}

\newcommand{\eigm}{\sigma^2_{\scriptsize \mbox{\upshape min}}}
\newcommand{\eigM}{\sigma^2_{\scriptsize \mbox{\upshape max}}}
\newcommand{\sigm}{\sigma_{\scriptsize \mbox{\upshape min}}}
\newcommand{\sigM}{\sigma_{\scriptsize \mbox{\upshape max}}}
\newcommand{\Tr}{ \mbox{\upshape Tr}}

\newcommand{\Eigm}{\varsigma_{\scriptsize \mbox{\upshape min}}}
\newcommand{\EigM}{\varsigma_{\scriptsize \mbox{\upshape max}}}

\textbf{Notation}: 
The set of real numbers is denoted $\Real$. 
Deterministic quantities are denoted using $a,\mat{a}$, or $\mat{A}$, where bold fonts denote vectors (\textit{i.e.}, $\mat{a}$) or matrices (\textit{i.e.}, $\mat{A}$). 
Random quantities are denoted using upper-case \emph{italics}, where $A$ is a random variable (RV), and $\pmb{A}$ a random vector/matrix.
Let $\Pr\{A\leq a\}$ denote the probability that event $\{A\leq a\}$ occurs.
Sets are denoted using braces, \textit{e.g.}, $\{1,2,\cdots\}$. 
The notation $\mathbb{E}$ denotes expectation.
The notation $i,j,\ell,\omega$ is used for indexing. 
We let $\Norm{\cdot}{p}$ denote the $\ell_p$-norm for $p=1$ and $2$. 

\newcommand{\x}{\alpha}
\newcommand{\xbk}{\overline{\pmb{\alpha}}_{k}}

\section{Preliminaries}  \label{sect:params}
\subsection{Compressed Sensing (CS) Theory} \label{ssect:params}

\newcommand{\Sens}{\pmb{\Phi}}
\renewcommand{\S}{\mathcal{S}}
\newcommand{\Fm}{F_{\sigma^2_{\scriptsize \mbox{\upshape min}}}}
\newcommand{\FM}{F_{\sigma^2_{\scriptsize \mbox{\upshape max}}}}
\renewcommand{\a}{a}
\newcommand{\matt}[1]{\pmb{#1}}
\newcommand{\Bas}{\mat{D}}

A vector $\mat{a}$ is said to be $k$-sparse, if at most $k$ vector coefficients are non-zero (\emph{i.e.}, its $\ell_0$-distance satisfies $\Norm{\mat{\a}}{0} \leq k$). 
Let $n$ be a positive integer that denotes {block length}, and let $\matt{\x}=[\x_1,\x_2,\cdots, \x_n]^T$ denote a length-$n$ signal vector with signal coefficients $\x_i$. 
The \emph{best $k$-term approximation} $\xbk $ of $\matt{\x}$, is obtained by finding the $k$-sparse vector $\xbk $ that has minimal approximation error $\Norm{\xbk - \matt{\x}}{2}$.

\newcommand{\col}{\pmb{\phi}}
\newcommand{\y}{b}
\newcommand{\eps}{\epsilon}


Let $\Sens$ denote an $m\times n$ CS sampling matrix, where $m <  n$. 
The length-$m$ \textbf{measurement vector} denoted $\mat{\y}=[\y_1,\y_2,\cdots, \y_m]^T $ of some length-$n$ signal $\matt{\x}$, is formed as $\mat{\y} =\Sens \matt{\x}$.
Recovering $\matt{\x}$ from $\mat{\y}$ is challenging as $\Sens$ possesses a \emph{non-trivial null-space}. 
We typically recover $\matt{\x}$ by solving the (convex) $\ell_1$-\textbf{minimization} problem
\bea
	  \min_{\tilde{\matt{\x}} \in \Real^n} \Norm{\tilde{\matt{\x}}}{1}~~~\mbox{ s. t.   } \Norm{\tilde{\mat{\y}} - \Sens \tilde{\matt{\x}}}{2} \leq \eps. \label{eqn:L1}
\eea
The vector $\tilde{\mat{\y}}$ is a \emph{noisy} version of the original measurements $\mat{\y}$, and here $\eps$ bounds the noise error, \emph{i.e.}, $\eps \geq \Norm{\tilde{\mat{\y}}-\mat{\y}}{2}$. 
Recovery conditions have been considered in many flavors~\cite{Donoho,DTExp,Grib,Elad,DT}, and mostly rely on studying parameters of the sampling matrix $\Sens$. 


\newcommand{\RIC}{\delta}

For $k \leq n$, the $k$-th \textbf{restricted isometry constant} $\RIC_k$ of an $m\times n$ matrix $\Sens$, equals the smallest constant that satisfies
\bea
	(1-\RIC_k) \Norm{\matt{\x}}{2}^2 \leq \Norm{\Sens\matt{\x}}{2}^2 \leq (1+\RIC_k) \Norm{\matt{\x}}{2}^2, \label{eqn:RIC}
\eea
for any $k$-sparse $\matt{\x} \mbox{ in } \Real^n$.
The following well-known recovery guarantee is stated w.r.t. $\RIC_k$ in (\ref{eqn:RIC}).


\newcommand{\ConstZero}{c_1}
\newcommand{\ConstOne}{c_2}
\newcommand{\half}{\frac{1}{2}}
\newcommand{\RIPthm}{A}
\begin{thmN}{\RIPthm, \textit{c.f.},~\cite{RIP}} 
Let $\Sens$ be the sensing matrix.
Let $\matt{\x}$ denote the signal vector. 
Let $\mat{\y}$ be the measurements, \emph{i.e.}, $\mat{\y} = \Sens\matt{\x}$. 
Assume that the $(2k)$-th restricted isometry constant $\RIC_{2k}$ of $\Sens$ satisfies $\RIC_{2k} < \sqrt{2} -1$, and further assume that the noisy version $\tilde{\mat{\y}}$ of $\mat{\y}$ satisfies $ \Norm{\tilde{\mat{\y}}-\mat{\y}}{2} \leq \epsilon$. Let $\xbk$ denote the best-$k$ approximation to $\matt{\x}$. Then the $\ell_1$-minimum solution $\matt{\x}^*$ to (\ref{eqn:L1}) satisfies 
\[
	 \Norm{\matt{\x}^* - \matt{\x}}{1} \leq \ConstZero \Norm{\matt{\x} - \xbk}{1} + \ConstOne \epsilon,
\]
for small constants $\ConstZero = 4\sqrt{1+\RIC_{2k}}/(1-\RIC_{2k}(1+\sqrt{2}))$ and $\ConstOne = 2 (\RIC_{2k} (1-\sqrt{2}) - 1)/(\RIC_{2k} (1+\sqrt{2}) - 1) $.
\end{thmN}

\newcommand{\Gram}{\Sens^T\Sens}

\newcommand{\ord}[1]{{(#1)}}
\newcommand{\GramS}{\Sens^T_\S\Sens_\S}
\newcommand{\SensS}{\Sens_\S}
\newcommand{\SensSi}{\Sens_{\S_i}}
\newcommand{\GramSi}{\Sens^T_{\S_i}\Sens_{\S_i}}
\newcommand{\GramSoi}[1]{\Sens_{\S_\ord{#1}}}
\newcommand{\GrampSoi}[1]{\Sens'_{\S_\ord{#1}}}
\newcommand{\N}{N}
\newcommand{\Bin}[2]{{#1 \choose #2}}
\newcommand{\define}{\stackrel{\Delta}{=}}
\renewcommand{\fn}{\footnote{We aim to relax this fairly restrictive assumption in future work.}}



\renewcommand{\fn}{\footnote{For simplicity, we omitted small deviation constants in Theorem B, see~\cite{Candes2004} p. 18 for details.}}

\newcommand{\A}{\pmb{A}}
Theorem \RIPthm~is very powerful, on condition that we know the constants $\RIC_k$. But because of their combinatoric nature, computing the restricted isometry constants $\RIC_k$ is NP-Hard~\cite{Blanchard}.
Let $\S$ denote a size-$k$ subset of indices. 
Let $\SensS$ denote the size $m\times k$ submatrix of $\Sens $, indexed on (column indices) in $\S$.  
Let $\eigM(\SensS)$ and $\eigm(\SensS)$ respectively denote the minimum and  maximum, \emph{squared-singular values} of $\SensS$. Then from (\ref{eqn:RIC}) if the columns $\col_i$ of $\Sens$ are properly normalized, \emph{i.e.}, if $\Norm{\col_i}{2}=1$, we deduce that $\RIC_{k}$ is the smallest constant in $\Real$ that satisfies 
\bea 
	\RIC_k &\geq& \max(\eigM(\Sens_\S) - 1, 1-\eigm(\Sens_\S)),
	\label{eqn:RIC2}
\eea
for all $\Bin{n}{k}$ size-$k$ subsets $\S$.
For large $n$, the number $\Bin{n}{k}$ is huge. 
Fortunately $\RIC_k$ need not be explicitly computed, if we can estimate it after incorporating \emph{randomization}~\cite{Candes,Donoho}. 

\newcommand{\func}{\zeta}
\newcommand{\1}{\mathbb{1}}
\newcommand{\LedThm}{B}

Recovery guarantee Theorem \RIPthm~involves ``worst-case'' analysis.
If the inequality (\ref{eqn:RIC2}) is violated for \emph{any} one submatrix $\Sens_\S$, then the \emph{whole} matrix $\Sens$ is deemed to have restricted isometry constant larger than $\RIC_k$. 
A common complaint of such ``worst-case'' analyses is pessimism, \textit{e.g.}, in~\cite{Dossal2009} it is found that for $n=4000$ and $m=1000$, the restricted isometry property is not even satisfied for sparsity $k=5$. 
This motivates the ``average-case'' analysis investigated here, where the recovery guarantee is relaxed to hold for a large ``fraction'' of signals (useful in applications that do not demand all possible signals to be completely recovered).
We draw ideas from the statistical StRIP notion used in deterministic CS, which only require ``most'' of the submatrices $\Sens_\S$ to satisfy some properties. 

In statistics, a well-known notion of a U-statistic (introduced in the next subsection) is very similar to StRIP. 
We will show how U-statistics naturally lead to ``average-case'' analysis.

%


\subsection{U-statistics \& StRIP} \label{ssect:UstatStr}

\newcommand{\E}{\mathbb{E}}
\newcommand{\g}{g}
\newcommand{\ZO}{\Real_{[0,1]}}

A function $\func : \Real^{m \times k} \rightarrow \Real$ is said to be a \textbf{kernel}, 
if for any $\mat{A},\mat{A}' \in \Real^{m \times k}$, we have $\func(\mat{A}) = \func(\mat{A}')$ if matrix ${\mat{A}}'$ can be obtained from $\mat{A}$ by \emph{column reordering}.
Let $\ZO$ be the set of real numbers bounded below by $0$ and above by $1$, i.e., $\ZO = \{a \in \Real : 0 \leq a \leq 1\}$.
U-statistics are associated with functions $g : \Real^{m \times k} \times \Real \rightarrow \ZO$ known as \textbf{bounded kernels}.
To obtain bounded kernels $g$ from indicator functions, simply use some kernel $\func$ and set 
$g(\mat{A},a) = \Ind{\func(\mat{A})\leq a}$ or $g(\mat{A},a) = \Ind{\func(\mat{A}) > a}$, \textit{e.g.} $\1\{\sigM^2(\mat{A}) \leq a\}$. 

\newcommand{\V}{U}
\begin{defn}[Bounded Kernel U-Statistics] \label{def:Ustat}
Let $\A$ be a random matrix with $n$ columns. Let $\Sens$ be sampled as $\Sens=\A$. 
Let $g : \Real^{m \times k} \times \Real \mapsto \ZO$ be a bounded kernel. For any $a\in \Real$, the following quantity
\bea
	 \U_n(a) \define \frac{1}{\Bin{n}{k}} \sum_{\S} g(\Sens_{\S}, a) \label{eqn:Ustat}
\eea
is a U-statistic of the sampled realization $\Sens=\A$, corresponding to the kernel $g$.
In 
(\ref{eqn:Ustat}), the matrix $\SensS$ is the submatrix of $\Sens$ indexed on column indices in $\S$, and the sum takes place over all subsets $\S$ in $\{1,2,\cdots, n\}$. 
Note, $0 \leq U_n(a) \leq 1$.
\end{defn}

For $k\leq n$ and positive $u$ where $u\leq 1$, a matrix $\Sens$ has $u$-\textbf{StRIP constant} $\RIC_k$, if $\RIC_k$ is the smallest constant s.t.
\bea
	(1-\RIC_k) \Norm{\matt{\x}}{2}^2 \leq \Norm{\Sens_\S\matt{\x}}{2}^2 \leq (1+\RIC_k) \Norm{\matt{\x}}{2}^2, 	\label{eqn:SRIC}
\eea
for any $\matt{\x}\in\Real^k $ and fraction $u$ of size-$k$ subsets $\S$. 
The difference between (\ref{eqn:SRIC}) and (\ref{eqn:RIC}) is that $\Sens_S$ is in place of $\Sens$.
This StRIP notion coincides with~\cite{Barga}.
Consider $\func(\mat{A}) = \max(\eigM(\mat{A}) -1,1 -\eigm(\mat{A})  )$ where here $\func$ is a kernel. 
Obtain a bounded kernel $g$ by setting $g(\mat{A},a) = \1\{\func(\mat{A}) > a \}$. 
Construct a U-statistic $\U_n(\RIC)$ of $\Sens$ the form $\U_n(\RIC)= \Bin{n}{k}^{-1}\sum_\S \1\{\func(\SensS) > \RIC \}$. 
Then if this U-statistic satisfies $U_n(\RIC)=1-u$, the $u$-StRIP constant $\RIC_k$ of $\Sens$ is at most $\RIC$, \emph{i.e.}, $\RIC_k \leq \RIC$.


\renewcommand{\fn}{\footnote{If $\mat{A}$ has full column rank, then $\mat{A}^\dagger = (\mat{A}^T\mat{A})^{-1} \mat{A}^T$,}}
To exploit apparent similarities between U-statistics and StRIP, we turn to two ``average-case'' guarantees found in the StRIP literature.
In the sequel, the conditions required by these two guarantees, 
will be analyzed in detail via U-statistics  - for now let us recap these guarantees.
First, an $\ell_1$-minimization recovery guarantee recently given in~\cite{Barga}, is a StRIP-adapted version of the ``worst-case'' guarantee Theorem \RIPthm.
For any non-square matrix $\mat{A}$, let $\mat{A}^\dagger$ denote the \textbf{Moore-Penrose pseudoinverse}\fn.
A vector $ \matt{\beta}$ with entries in $\{-1,1\}$ is termed a \textbf{sign vector}.
For $\matt{\x}\in \Real^n$, we write $\matt{\x}_\S$ for the length-$k$ vector supported on $\S$.
Let $\S_c$ denote the complementary set of $\S$, \emph{i.e.}, $\S_c = \{1,2,\cdots, n\}\setminus \S$.
The ``average-case'' guarantees require us to check conditions on $\Sens$ for fractions of subsets $\S$, or \textbf{sign-subset} pairs $(\matt{\beta},\S)$.

\newcommand{\BargThm}{B}
\begin{thmN}{\BargThm,~\textit{c.f.},~Lemma 3,~\cite{Barga}}
Let $\Sens$ be an $m\times n$ sensing matrix. Let $\S$ be a size-$k$ subset, and let $ \matt{\beta} \in \{-1,1\}^k$. Assume that $\Sens$ satisfies 
\bitm
\item invertibility: for at least a fraction $1-u_1$ of subsets $\S$, the condition $\sigm(\SensS) > 0$ holds. 
\item small projections: for at least a fraction $1-u_2$ of sign-subset pairs $(\matt{\beta},\S)$, the condition 
\bea
	 \left|(\Sens_\S^\dagger \col_i)^T \matt{\beta} \right| \leq a_2  \mbox{ for every } i \notin \S\nonumber
\eea
holds where we assume the constant $a_2 < 1$. 
\item worst-case projections: for at least a fraction $1-u_3$ of subsets $\S$, the following condition holds
\bea
	 \Norm{\Sens_\S^\dagger \col_i}{1} \leq a_3  \mbox{ for every } i \notin \S. \nonumber
\eea
\eitm
Then for a fraction $1-u_1-u_2-u_3$ of sign-subset pairs $(\matt{\beta},\S)$, the following error bounds are satisfied
\bea
	 \Norm{\matt{\x}^*_\S - \matt{\x}_\S}{1} &\leq& \frac{2a_3}{1-a_2} \Norm{\matt{\x} - \xbk}{1}, \nn
	 \Norm{\matt{\x}^*_{\S_c} - \matt{\x}_{\S_c}}{1} &\leq& \frac{2 }{1-a_2} \Norm{\matt{\x} - \xbk}{1}, \nonumber
\eea
where $\matt{\x}$ is a signal vector that satisfies $\sgn(\matt{\x}_\S)= \matt{\beta}$, and $\xbk$ is the best-$k$ approximation  of $\matt{\x}$ and $\xbk$ is supported on $\S$, and finally $\matt{\x}^*$ is the solution to (\ref{eqn:L1}) where the measurements $\mat{\y}$ satisfy $\mat{\y} = \Sens\matt{\x}$. 
\end{thmN}

\newcommand{\Reg}{\theta_n}
\newcommand{\noisesd}{c_Z}


For convenience, the proof is provided in Supplementary Material \ref{app:recov}.
The second guarantee is a StRIP-type recovery guarantee for the \emph{LASSO} estimate, based on~\cite{Cand2008} (also see~\cite{Barga}). 
Consider recovery from noisy measurements
\[
	\tilde{\mat{\y}} = \Sens \matt{\x}+ \mat{z},
\]
here $\mat{z}$ is a length-$m$ noise realization vector. We assume that the entries $z_i$ of $\mat{z}$, are sampled from a zero-mean Gaussian distribution with variance $\noisesd^2$.
The LASSO estimate considered in~\cite{Cand2008}, is the optimal solution $\matt{\x}^*$ of the optimization problem
\bea
	  \min_{\tilde{\matt{\x}} \in \Real^n}  \frac{1}{2} \Norm{\tilde{\mat{\y}} - \Sens \tilde{\matt{\x}}}{2}  +  2 \noisesd \cdot\Reg  \Norm{\tilde{\matt{\x}}}{1}. \label{eqn:Lasso}
\eea
The $\ell_1$-regularization parameter is chosen as a \emph{product} of two terms $\noisesd$ and $ \Reg$, where we specify $\Reg=(1+a)\sqrt{2 \log n }$ for some positive $a$.
What differs from convention is that the regularization depends on the noise standard deviation $\noisesd$.
We assume $\noisesd > 0$, otherwise there will be no $\ell_1$-regularization.

\newcommand{\CandThm}{C}
\begin{thmN}{\CandThm, \textit{c.f.},~\cite{Cand2008}}
Let $\Sens$ be the $m\times n$ sensing matrix. Let $\S$ be a size-$k$ subset, and let $ \matt{\beta} \in \{-1,1\}^k$. 
\bitm
\item invertability: for at least a fraction $1-u_1$ of subsets $\S$, the condition $\sigm(\SensS) > {a_1}$ holds. 
\item small projections: for at least a fraction $1-u_2$ of subsets $\S$, same as Theorem \BargThm.
\item invertability projections: for at least a fraction $1-u_3$ of sign-subset pairs $(\matt{\beta},\S)$, the following condition holds
\bea
	 \Norm{(\SensS^T\SensS)^{-1}\matt{\beta}}{\infty} \leq a_3. \nonumber
\eea
\eitm
Let $\noisesd$ denote noise standard deviation.
Assume Gaussian noise realization $\mat{z}$ in measurements $\mat{\tilde{\y}}$, satisfy 
\bitm
\item[i)] $\Norm{(\Sens_\S^T \Sens_\S)^{-1}\Sens_\S^T \mat{z}}{\infty} \leq (\noisesd \sqrt{2\log n})/a_1$, for the constant $a_1$ in the invertability condition.
\item[ii)] $\Norm{\Sens_{\S_c}^T (\mat{I} - \Sens_\S \Sens_\S^\dagger)\mat{z}}{\infty} \leq \noisesd 2\sqrt{\log n}$, where $\S_c$ is the complementary set of $\S$.
\eitm
For some positive $a$, assume that constant $a_2$ in the small projections condition, satisfies
\bea
(\sqrt{2}(1+a))^{-1} + a_2 < 1.\label{eqn:Cand1}
\eea
Then for a fraction $1-u_1-u_2-u_3$ of sign-subset pairs $(\matt{\beta},\S)$, the LASSO estimate $\matt{\x}^*$ from (\ref{eqn:Lasso}) with regularization $\Reg=(1+a)\sqrt{2 \log n }$ for the same $a$ above, will successfully recover both signs and supports of $\matt{\x}$, if
\bea
   |\alpha_i| \geq  \left[a_1^{-1}  + 2 a_3 (1+a) \right] \cdot  \noisesd\sqrt{2 \log n}~\mbox{ for all }~i \in \S \label{eqn:Cand2}
\eea
\end{thmN}

Because of some differences from~\cite{Cand2008}, we also provide the proof in Supplementary Material \ref{app:recov}.
In~\cite{Cand2008} it is shown that the noise conditions i) and ii) are satisfied with large probability at least $1 - n^{-1} (2\pi \log n)^{-\half}$ (see Proposition \ref{pro:noise} in Supplementary Material \ref{app:recov}).
Theorem \CandThm~is often referred to as a \emph{sparsity pattern recovery} result, in the sense that it guarantees recovery of the sign-subset pairs $(\matt{\beta}, \S)$ belonging to a $k$-sparse signal $\matt{\x}$. 
Fuchs established some of the earlier important results, see~\cite{Fuchs,Tropp2,Fuchs2005}.   

In Theorems \BargThm~and \CandThm, 
observe that the \emph{invertability} condition can be easily checked using an U-statistic; simply set the bounded kernel $g$ as $g(\mat{A},a_1) = \I{\sigm(\mat{A}) \leq a_1}$ for some positive $a_1$ and measure the fraction $U_n(a_1) = u_1$. 
Other conditions require slightly different kernels, to be addressed in upcoming Section \ref{sect:KKT}. 
But first we first introduce the main U-statistical large deviations theorem (central to our analyses) in the next section.



\section{Large deviation theorem: ``average-case'' behavior} \label{sect:Dist}

Consider two bounded kernels $g$ defined for $\mat{A}\in \Real^{m\times k}$, corresponding to maximum and minimum squared singular values
\bea
g(\mat{A}, a) &=& \I{\eigM(\mat{A}) \leq a}, \mbox{ and } \label{eqn:geig_max} \\
g(\mat{A}, a) &=& \I{\eigm(\mat{A}) \leq a}.  \label{eqn:geig_min}
\eea
Note that restricted isometry conditions (\ref{eqn:RIC}) and (\ref{eqn:SRIC}) depend on both $\eigm$ and $\eigM$ behaviors, although
the conditions in the previous StRIP-recovery guarantees Theorem \BargThm~are explicitly imposed only on $\eigm$.
See~\cite{Blanchard,Edelman} for the different behaviors and implications of these two extremal eigenvalues.
In this section we consider two U-statistics, corresponding separately to (\ref{eqn:geig_max}) and (\ref{eqn:geig_min}).

\newcommand{\gauss}{\int_{-\infty}^a (1/\sqrt{2\pi}) e^{-\t^2/2}d\t}
\renewcommand{\t}{t}

\ifthenelse{\boolean{dcol}}{
\begin{figure}[!t]
	\centering
	  \epsfig{file={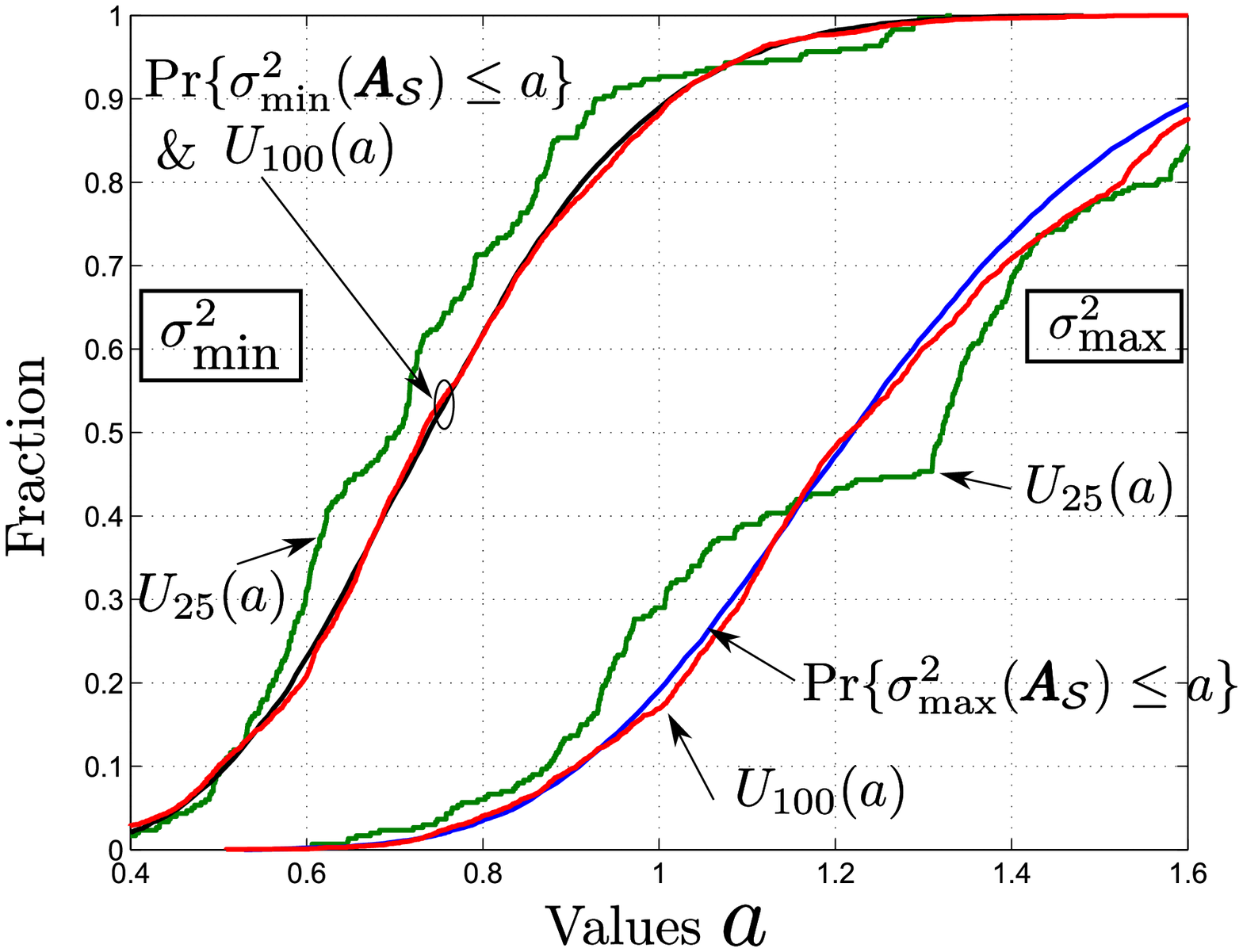},width=.6\linewidth}
		\caption{Gaussian measure. Concentration of U-statistic $\U_n(a)$ for squared singular value $\eigm$ and $\eigM$ kernels $g$, see (\ref{eqn:geig_max}) and (\ref{eqn:geig_min}). Shown for $m=25, k=2$ and two values of $n = 25$ and $100$.}
		\label{fig:Intro}
		\vspace*{-10pt}
\end{figure}
}{
\begin{figure}[!t]
	\centering
	  \epsfig{file=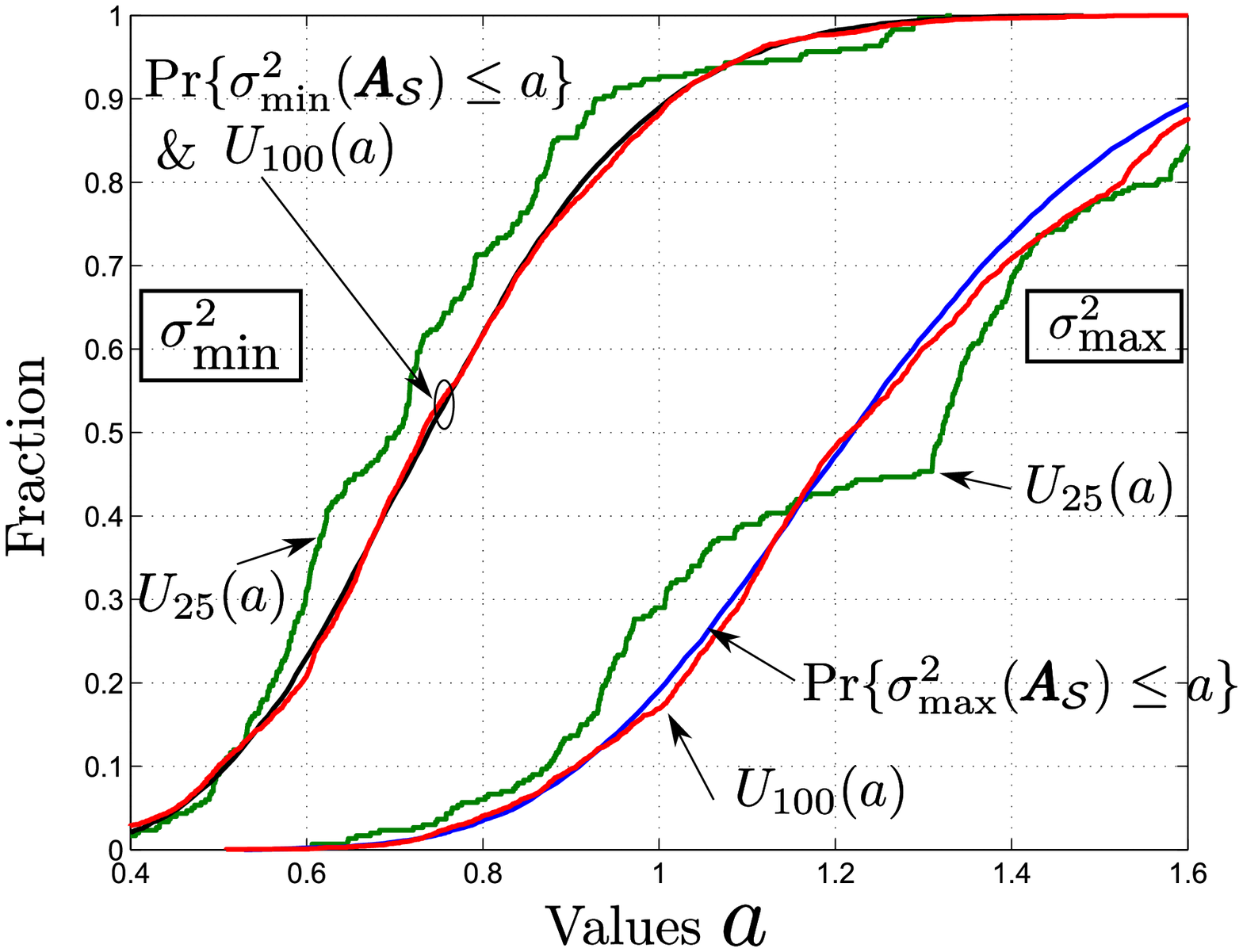},width=.5\linewidth}
		\caption{Gaussian measure. Concentration of U-statistic $\U_n(a)$ for squared singular value $\eigm$ and $\eigM$ kernels $g$, see (\ref{eqn:geig_max}) and (\ref{eqn:geig_min}). Shown for $m=25, k=2$ and two values of $n = 25$ and $100$.}
		\label{fig:Intro}
		\vspace*{-10pt}
\end{figure}
}

Let $\A_i$ denote the $i$-th column of $\A$, and assume $\A_i$ to be IID. 
For an bounded kernel $g$, let $p(a)$ denote the expectation $\E g(\A_\S,a)$, \textit{i.e.}, $p(a) = \E g(\A_\S,a)$ for any size-$k$ subset $\S$. 
Since $p(a) = \E \U_n(a)$, thus the U-statistic mean $\E \U_n(a)$ does not depend on block length $n$.

\begin{thm} \label{thm:Ustat}
Let $\A$ be an $m\times n$ random matrix, whereby the columns $\A_i$ are IID. Let $g$ be a bounded bounded kernel that maps $\Real^{m\times k} \times \Real \rightarrow \ZO$ and let $p(a)=\E g(\A_\S,a) = \E \V_n(a)$. Let $\V_n(a)$ be a U-statistic of the sampled realization $\Sens=\A$ corresponding to the bounded kernel $g$. Then almost surely when $n$ is sufficiently large, the deviation $|\V_n(a) - p(a)| \leq \eps_n(a)$ is bounded by an error term $\eps_n(a)$ that satisfies 
\bea
	\eps_n^2(a) =  2 p(a) (1-p(a)) \cdot (n/k)^{-1} \log (n/k) . \label{eqn:UstatThm}
\eea
\end{thm}
Theorem \ref{thm:Ustat} is shown by piecing together (5.5) in~\cite{Hoef} and Lemma 2.1 in~\cite{Sen}. The proof is given in Appendix \ref{app:proofDev}. Figure \ref{fig:Intro} empirically illustrates this concentration result for $g$ in (\ref{eqn:geig_max}) and (\ref{eqn:geig_min}), 
corresponding to $p(a) = \E g(\A_\S,a) = \Pr\{\eigM(\A_\S) \leq a\}$ and $p(a) = \Pr\{\eigm(\A_\S) \leq a\}$. Empirical simulation of restricted isometries is very difficult, thus we chose small values $k=2$, $m=25$ and block lengths $n=25$ and $n=100$. For $n=25$ the deviation $|U_{25}(a) - p(a)|$ is very noticeable for all values of $a$ and both $\eigM$ and $\eigm$. However for larger $n=100$, the deviation $|U_{100}(a) - p(a)|$ clearly becomes much smaller. 
This is predicted by vanishing error $\eps_n(a)$ given in Theorem \ref{thm:Ustat}, which drops as the ratio $n/k$ increases.
In fact if $k$ is kept constant then the error behaves as $\mathcal{O}(n^{-1}\log n )$.

\ifthenelse{\boolean{dcol}}{ 
\begin{figure}[!t]
	\centering
	  \epsfig{file={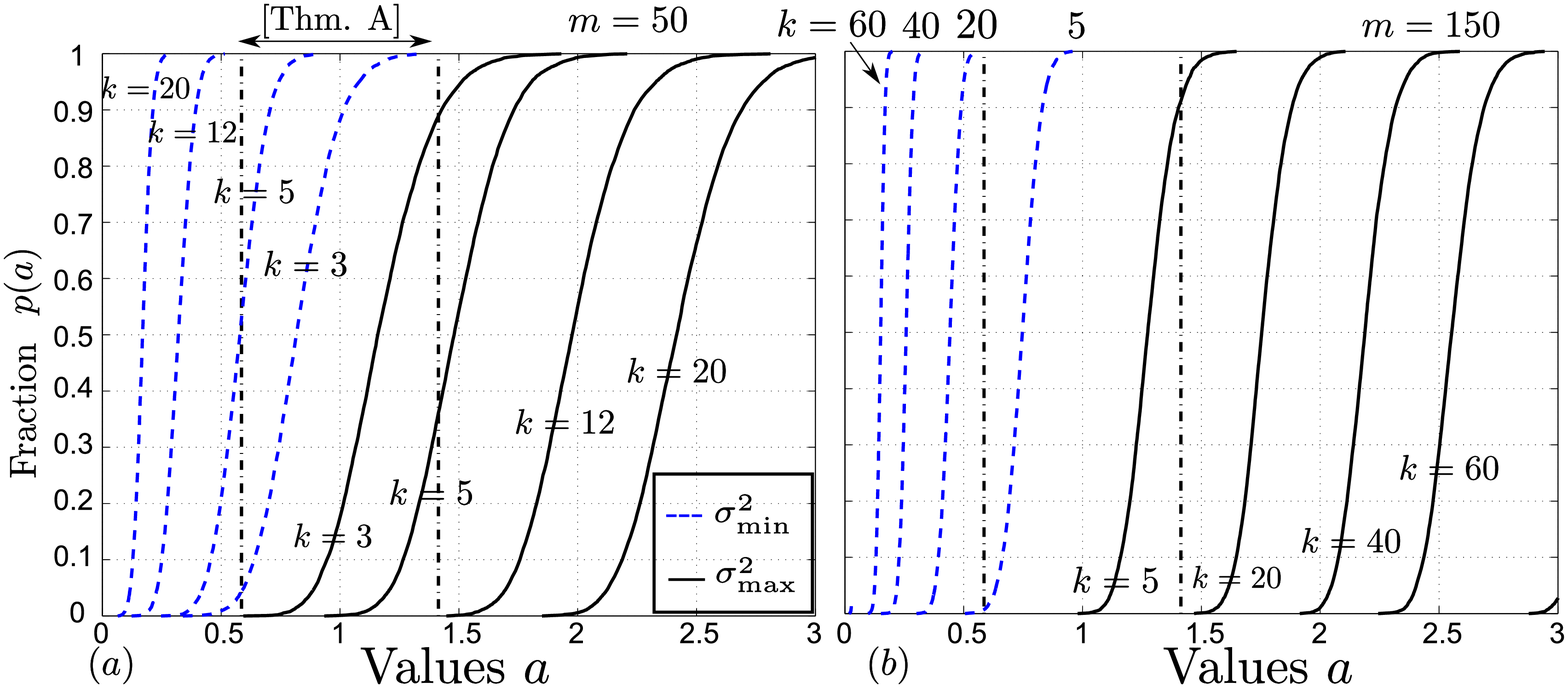},width=1\linewidth}
		\caption{Means $p(a) = \E \U_n(a)$ for predicting the concentration of $\U_n(a)$. Shown for the Gaussian case, $(a)$ $m=50$ and $(b)$ $m=150$.}
		\label{fig:Gauss}
\end{figure}
}{
\begin{figure}[!t]
	\centering
	  \epsfig{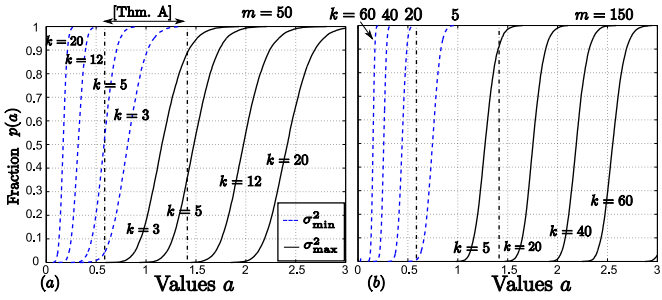}
		\caption{Means $p(a) = \E \U_n(a)$ for predicting the concentration of $\U_n(a)$. Shown for the Gaussian case, $(a)$ $m=50$ and $(b)$ $m=150$.}
		\label{fig:Gauss}
\end{figure}
}




\renewcommand{\fn}{\footnote{We point out that Bah actually defined two separate restricted isometry constants, each corresponding to $\eigm$ and $\eigM$ in~\cite{Bah}. In this paper to coincide the presentation with our discussion on squared singular values, their results will be discussed in the domain of $\eigm$ and $\eigM$.}}

\newcommand{\fnT}{\footnote{The analysis in~\cite{Bah} was performed for the large limit of $k,m$ and $n$, where both $k/m$ and $m/n$ approach fixed constants.}}


Table \ref{tab:1} reproduces\fn~a sample of (asymptotic) estimates for both $\eigM$ and $\eigm$ cases, taken from~\cite{Bah}. 
These estimates are derived for ``worst-case'' analysis,
under assumption that every entry $A_{ij}$ of $\pmb{A}$ is IID and Gaussian distributed (\textit{i.e.}, $A_{ij}$ is Gaussian with variance $1/m$). 
Table \ref{tab:1} presents the estimates according\fnT~to fixed ratios $k/m$ and $m/n$. 
To compare, Figure \ref{fig:Gauss} shows the expectations $p(a)=\E \U_n(a)$. 
The values $p(a)$ are interpreted as fractions, and as $n/k$ becomes large $p(a)$ is approached by $\U_n(a)$ within a stipulated error $\eps_n$. 
Figure \ref{fig:Gauss} is empirically obtained, though note that in Gaussian case
for $p(a)$ we also have exact expressions~\cite{Edelman,Koev}, and the \emph{Bartlett decomposition}~\cite{Edelman2005}, available.
Again $p(a)$ is a marginal quantity (\textit{i.e.} does not depend on $n$) 
and simulation is reasonably feasible. 
In the spirit of non-asymptotics, we consider relatively small $k,m$ values as compared to other works~\cite{Bah,Dossal2009}; these adopted values are nevertheless ``practical'', in the sense they come an implementation paper~\cite{Fred}.

\begin{table}[t]
	\centering
\caption{Asymptotic Lower and Upper Bounds on ``Worst-Case'' Eigenvalues,~\cite{Bah}}
		\begin{tabular}{|@{\hspace{.3ex}}c@{\hspace{.3ex}}|c|c|c|c||c|c|c|}
			\cline{1-8}
			    &   & \multicolumn{3}{|c||}{Minimum: $\sigma^2_{\mbox{\scriptsize min}}$} 
			        & \multicolumn{3}{|c|}{Maximum: $\sigma^2_{\mbox{\scriptsize max}}$} \\\cline{3-8}
			    &   & \multicolumn{3}{|c||}{$m/n$} & \multicolumn{3}{|c|}{$m/n$} \\\cline{3-8}
					&   & \emph{0.1} &  \emph{0.3} & \emph{0.5} & \emph{0.1} & \emph{0.3} & \emph{0.5}\\\hline 
			 \multirow{3}{*}{\begin{sideways}$k/m$ \end{sideways}} 
			        & \emph{0.1} & 0.095	& 0.118	& 0.130	& 3.952	& 3.610	& 3.459 \\\cline{2-8}
			        & \emph{0.2} & 0.015	& 0.026	& 0.034	& 5.587	& 4.892	& 4.535 \\\cline{2-8}
			        & \emph{0.3} & 0.003	& 0.006	& 0.010	& 6.939	& 5.806	& 5.361 \\\hline
		\end{tabular}
		\label{tab:1}
\end{table}

Differences are apparent from comparing ``average-case'' (Figure \ref{fig:Gauss}) and ``worst-case'' (Table \ref{tab:1}) behavior. Consider $k/m = 0.3$ where Table \ref{tab:1} shows for all undersampling ratios $m/n$, the worst-case estimate of $\eigm$ is very small, approximately $0.01$.
But for fixed $m=50$ and $m=150$, Figures \ref{fig:Gauss}$(a)$ and $(b)$ show that for respectively $k= 0.3\cdot (150)=15$ and $k=45$, a large fraction of subsets $\S$ seem to have $\eigm(\Sens_\S)$ lying above $0.1$. 
From Table \ref{tab:1}, the estimates for $\eigm$ gets worse (\textit{i.e.}, gets smaller) as $m/n$ decreases.
But the error $\eps_n(a)$ in Theorem \ref{thm:Ustat} vanishes with larger $n/k$.
For the other $\eigM$ case, we similarly observe that the values in Table \ref{tab:1} also appear more ``pessimistic''.

We emphasize that Theorem \ref{thm:Ustat} holds regardless of distribution. Figure \ref{fig:BernUnif} is the counterpart figure for Bernoulli and Uniform cases (\textit{i.e.}, each entry $A_{ij}$ is respectively drawn uniformly from $\{-1/\sqrt{m},1/\sqrt{m}\}$, or $\{a \in \Real : |a| \leq \sqrt{3/m}\}$), shown for $m=50$. Minute differences are seen when comparing with previous Figure \ref{fig:Gauss}.
For $k=3$, we observe the fraction $p(a)$ corresponding to $\eigM$ to be roughly 0.95 in the latter case, whereas in the former we have roughly $0.9$ in Figure \ref{fig:BernUnif}$(a)$, and $0.88$ in Figure \ref{fig:BernUnif}$(b)$. 



\begin{rem} \label{rem:1}
Exponential bounds on $\Pr\{\min_\S \sigm^2(\pmb{A}_\S) < 1-\RIC\}$ and $\Pr\{\max_\S \sigM^2(\pmb{A}_\S) > 1+\RIC\}$ for $\max(\RIC,\sqrt{k/m}) < \sqrt{2}-1$, see (\ref{eqn:RIC2}), employed in ``worst-case'' analyses, give the optimal $m  = \mathcal{O}(k\log(n/k))$ rate, see~\cite{Candes,Rudelson2010,RIPsimp}. 
However the implicit constants are inherently not too small (\textit{i.e.}, these constants cannot be improved). 
\end{rem}


These comparisons motivate ``average-case'' analysis. Marked out on Figures \ref{fig:Gauss} and \ref{fig:BernUnif} are the ranges for which $\eigM$ and $\eigm$ must lie to apply Theorem \RIPthm~(``worst-case'' analysis). In the cases shown above, the observations are somewhat disappointing - even for small $k$ values, a substantial fraction of eigenvalues lie outside of the required range. 
Thankfully, there exist ``average-case'' guarantees, \textit{e.g.}, previous Theorems~\BargThm~and~\CandThm, addressed in the next section.

\ifthenelse{\boolean{dcol}}{ 
\begin{figure}[!t]
	\centering
	  \epsfig{file={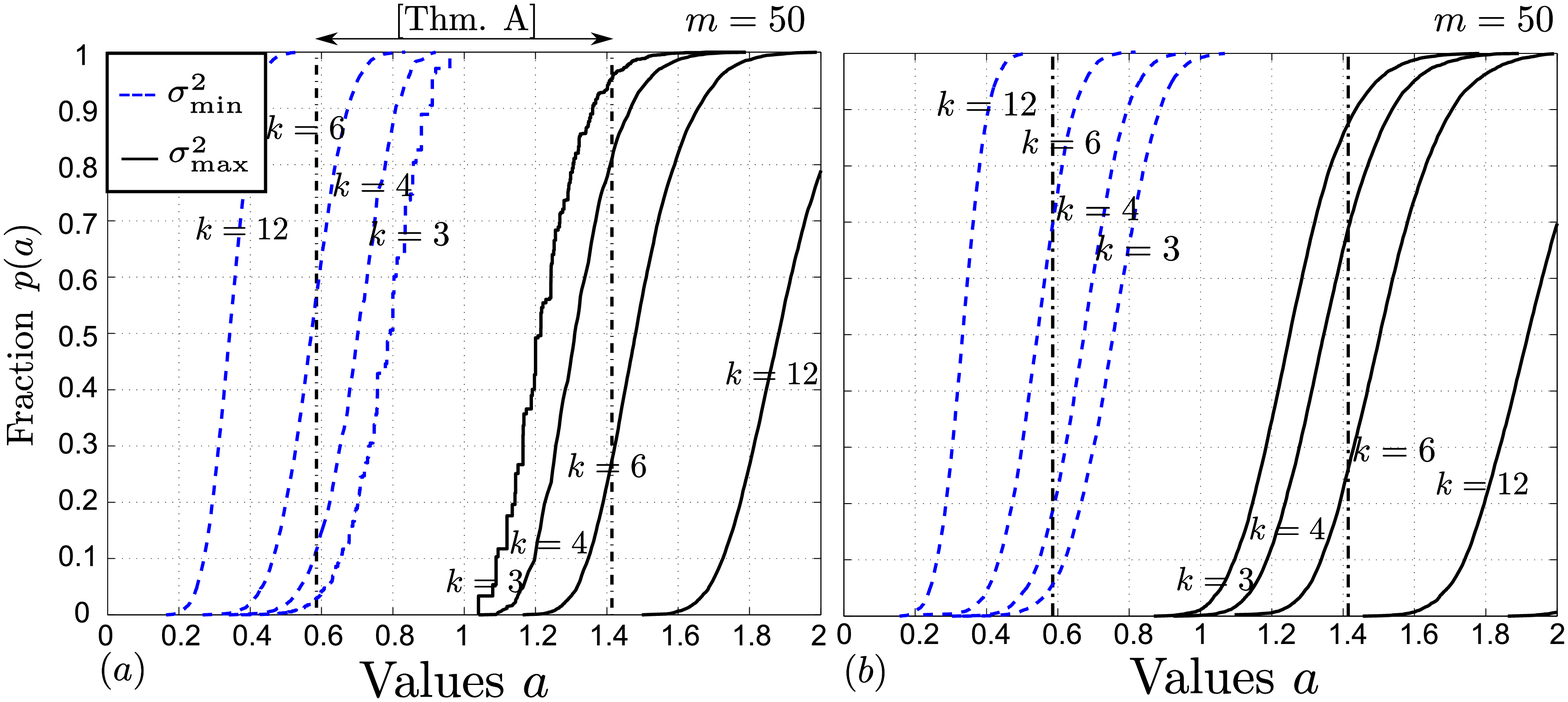},width=1\linewidth}
		\caption{Means $p(a) = \E \U_n(a)$ for $m=50$ and the $(a)$ Bernoulli and $(b)$ Uniform cases.}
		\label{fig:BernUnif}
		\vspace*{-15pt}
\end{figure}
}{
\begin{figure}[!t]
	\centering
	  \epsfig{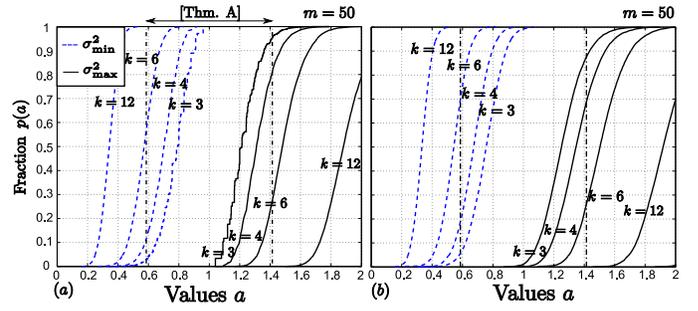},width=.8\linewidth}
		\caption{Means $p(a) = \E \U_n(a)$ for $m=50$ and the $(a)$ Bernoulli and $(b)$ Uniform cases.}
		\label{fig:BernUnif}
		\vspace*{-15pt}
\end{figure}
}

\newcommand{\mtail}{p}
\newcommand{\jtail}{q}
\newcommand{\z}{a}
\newcommand{\Lam}{\lambda_n}
\newcommand{\kernel}{\zeta}

\renewcommand{\define}{=}
\newcommand{\PoiThm}{2.N}

\newcommand{\Iv}{^{-1}}
\newcommand{\R}{\mathcal{R}}
\section{U-statistics \& ``Average-case'' Recovery Guarantees} \label{sect:KKT}

\newcommand{\s}{\pmb{\beta}}
\newcommand{\jp}{\ell}
\newcommand{\Rj}{{\R\setminus \{j\}}}

\subsection{Counting argument using U-statistics} \label{ssect:count}

Previously we had explained how the \emph{invertability} conditions required by Theorems \BargThm~and \CandThm~naturally relate to U-statistics. 
We now go on to discuss the other conditions,  whereby the relationship may not be immediate.
We begin with the \emph{projections} conditions, in particular 
the \emph{worst-case projections} condition.
For given $\Sens$, we need to \emph{upper bound} the fraction of subsets $\S$, for which there exists \emph{at least one} column $\col_j$ where $j\notin\S$, such that $\Norm{\Sens_{\S}^\dagger \col_j }{\infty}$ exceeds some value $a$. 
To this end, let $\R$ denote a size-$(k+1)$ subset, and $\Rj$ is the size-$k$ subset excluding the index $j$.
Consider the bounded kernel $g : \Real^{m \times (k+1)} \times \Real \mapsto \ZO$ set as
\bea
g(\mat{A},a) = \frac{1}{k+1}   \sum_{j=1}^{k+1}  \I{ \Norm{\mat{A}_{\Rj}^\dagger \mat{a}_j}{\infty}> a},
\label{eqn:g2}
\eea
where here $\R = \{1,2,\cdots, k+1\}$, and $\mat{a}_j$ denotes the $j$-th column of $\mat{A}$. 
Consider the U-statistic with bounded kernel (\ref{eqn:g2}). We claim that  
\ifthenelse{\boolean{dcol}}{ 
\begin{align}
(n&-k)\cdot U_n(a) \nn
&= \frac{n-k}{(k+1)\Bin{n}{k+1}} \sum_{\R} \!  \sum_{j\in \R}  \I{ \Norm{\Sens_\Rj^\dagger \col_j}{\infty} > a}, \nn
&= \frac{1}{\Bin{n}{k}} \sum_{\S} \!  \sum_{j\notin \S}  \I{ \Norm{\Sens_\S^\dagger \col_j}{\infty} > a}, \nonumber
\end{align}}{ 
\bea
(n-k)\cdot U_n(a) 
&=& \frac{n-k}{(k+1)\Bin{n}{k+1}} \sum_{\R} \!  \sum_{j\in \R}  \I{ \Norm{\Sens_\Rj^\dagger \col_j}{\infty} > a}, \nn
&=& \frac{1}{\Bin{n}{k}} \sum_{\S} \!  \sum_{j\notin \S}  \I{ \Norm{\Sens_\S^\dagger \col_j}{\infty} > a}, \nonumber
\eea}
where the summations over $\R$ and $\S$ are over all size-$(k+1)$ subsets, and all size-$k$ subsets, respectively.
The first equality follows from Definition \ref{def:Ustat} and (\ref{eqn:g2}). 
The second equality requires some manipulation.
First the coefficient $\Bin{n}{k}\Iv$ follows from the binomial identity $ \Bin{n}{k+1}\cdot (k+1) =  \Bin{n}{k}\cdot (n-k)$.
Next for some subset $\S$ and index $j$, write the indicator $\I{ \Norm{\Sens_\S^\dagger \col_j}{\infty} > a}$ as $\1_{\S,j}$ for brevity's sake.
By similar counting that proves the previous binomial identity, we argue $\sum_\R \sum_{j\in \R} \1_{\Rj,j} = \sum_\S \sum_{j\notin \S}\1_{\S,j}$, which then proves the claim.
Imagine a grid of ``pigeon-holes'', indexed by pairs $(\S,j)$, where $j \notin\S$. 
For each size-$(k+1)$ subset $\R$, we assign $k+1$ indicators $\1_{\Rj,j}$ to $k+1$ pairs $(\S,j)$.
No ``pigeon-hole'' gets assigned more than once.
In fact we infer from the binomial identity, that every ``pigeon-hole'' is in fact assigned exactly once, and argument is complete.

Similarly for the \emph{small projections} condition, we define a different bounded kernel $g : \Real^{m \times (k+1)} \times \Real \mapsto \ZO$ as
\bea
g(\mat{A},a) = \frac{1}{2^k(k+1)} \sum_{\jp=1}^{2^k}  \sum_{j=1}^{k+1}   \I{ \left|(\mat{A}_{\Rj}^\dagger \mat{a}_j)^T\s_\jp \right|> a},
\label{eqn:g3}
\eea
where $\R = \{1,2,\cdots, k+1\}$, and $\mat{a}_j$ denotes the $j$-th column of $\mat{A}$, and $\s_1,\s_2,\cdots,\s_{2^k}$ enumerate all $2^k$ unique sign-vectors in the set $\{-1,1\}^k$.
By similar arguments as before, we can show for the U-statistic $U_n(a)$ of $\Sens$ corresponding to the bounded kernel (\ref{eqn:g3}) satisfies
\ifthenelse{\boolean{dcol}}{
\bea
(n-k)\cdot U_n(a) 
\!\!\! &=& 
\!\!\!\frac{1}{2^k \Bin{n}{k}} \sum_{\jp=1}^{2^k} \sum_{\S} \!  \sum_{j\notin \S}  \I{\left|(\Sens_\S^\dagger \col_j)^T \s_\jp\right| > a}, \nonumber
\eea}{
\bea
(n-k)\cdot U_n(a) 
&=& \frac{1}{2^k \Bin{n}{k}} \sum_{\jp=1}^{2^k} \sum_{\S} \!  \sum_{j\notin \S}  \I{\left|(\Sens_\S^\dagger \col_j)^T \s_\jp\right| > a}, \nonumber
\eea}
For indicators $\1_{\S,j}$, note that $\sum_{j\notin \S} \1_{\S,j} \geq 1$ if \emph{at least one} indicator satisfying $\1_{\S,j} =1$, and we proved the following.
%
%

\begin{pro} \label{cor:count}
Let $\V_n(a_3)$ be the U-statistic of $\Sens$, corresponding to the bounded kernel $g(\mat{A},a_3)$ in (\ref{eqn:g2}).  
Then the fraction of subsets $\S$ of size-$k$, for which the worst-case projections condition is violated for some $a_3\in \Real$, is at most $(n-k)\cdot U_n(a_3)$. 
Similarly if $\V_n(a_2)$ corresponds to $g(\mat{A},a_2)$ in (\ref{eqn:g3}), the fraction sign-subset pairs $(\s,\S)$, for which the small projections condition is violated for some $a_2\in \Real$, is at most $(n-k)\cdot U_n(a_2)$. 
\end{pro}

Referring back to Theorem \BargThm, we point out that the \emph{small projections} condition is more stringent than the \emph{worst-case projections} condition. We mean the following:
in the former case, the value $a_2$ must be chosen such that $a_2 < 1$;  
in the latter case, the value $a_3$ is allowed to be larger than $1$, its size only affects the constant $2 a_3/(1-a_2)$ appearing in the error estimate $\Norm{\matt{\x}_\S^*-\matt{\x}_\S}{1}$.
In fact if the signal $\matt{\x}$ is $k$-sparse, then $\Norm{\matt{\x} - \xbk}{1} = 0$ and the size of $a_3$ is inconsequential, \textit{i.e.}, the \emph{worst-case projections} condition is not required in this special case.
In this special case, it is best to set $a_2 = 1-\eps$ for some arbitrarily small $\eps$.
Theorem \BargThm~is in fact a stronger version of Fuchs' early work on $\ell_0/\ell_1$-\emph{equivalence}~\cite{Fuchs}.
In the same respect, Donoho \& Tanner also produced early seminal results from counting faces of random polytopes~\cite{DTExp,DT}.

Figure \ref{fig:CompL1} shows empirical evidence, where the $k,m,n$ values are inspired by practical system sizes taken from an implementation paper~\cite{Fred}. 
These experiments consider $\Sens$ sampled from Gaussian matrices $\pmb{A}$, \emph{exactly} $k$-sparse signals with non-zero $\alpha_i$ sampled from $\{-1,1\}$, and uses $\ell_1$-minimization recovery (\ref{eqn:L1}).
Figure \ref{fig:CompL1}$(a)$ plots simulated (sparsity pattern recovery) results for 3 measurement sizes $m=50,100$ and $150$ and block sizes $n \geq 200$ and $n \leq 3000$. 
For example the contour marked ``0.1'', delineates the $k,n$ values for which recovery fails for a 0.1 fraction of (random) sparsity patterns (sign-subset pairs $(\s,\S)$).  
We examine the U-statistic $U_n(a_2)$ with kernel (\ref{eqn:g3}), related to the small projections condition.
Since $\pmb{A}$ has Gaussian distribution, we set $a_2=1$ in the kernel $g(\mat{A},a_2)$, as $\Pr\{(\pmb{A}_\S^\dagger\pmb{A}_i)^T\s = 1\}=0$ for any $(\s,\S)$ and $j\notin\S$.
Figure \ref{fig:CompL1}$(b)$ plots the expectation $(n-k)\cdot p(1)$, where $p(1) = \E U_n(1)= \E g(\pmb{A}_\R,1)$ for any size-$(k+1)$ subset $\R$. 
Again the contour marked ``0.1'', delineates the $k,n$ values for which $(n-k)\cdot p(1)= 0.1$.
Here the values $p(1)$ are empirical.
We observe that both Figures \ref{fig:CompL1}$(a)$ and $(b)$ are remarkably close for fractions $0.5$ and smaller. 
Figures \ref{fig:CompL1}$(c)$ incorporates the large deviation error $\eps_n$ given in Theorem \ref{thm:Ustat} (in doing so, we assume $n$ sufficiently large). 
The bound is still reasonably tight for fractions $\leq 0.5$. 
Comparing with recent Donoho \& Tanners' (also ``average-case'') results for $\ell_1$-recovery (for only the noiseless case), taken from~\cite{DTExp}.
For fractions $0.5$ and $0.01$, we observe that for system parameters $m=50$ and $n\leq 1000$ (chosen in hardware implementation~\cite{Fred}), we do not obtain reasonable predictions. For $m=100$, the bounds~\cite{DTExp} work only for very small block lengths $n\leq 300$. The only reasonable case here is $m=150$, where the bounds~\cite{DTExp} perform better than ours only for lengths $n\leq 400$ (\textit{i.e.}, Figure \ref{fig:CompL1}$(c)$ shows that for $n=300$, the large deviation bounds predict a 0.01 fraction of size $k=5$ unrecoverable sparsity patterns, but~\cite{DTExp} predict a 0.01 fraction of size $k=11$ unrecoverable sparsity patterns).

\begin{figure*}[!t]
	\centering
	  \epsfig{file={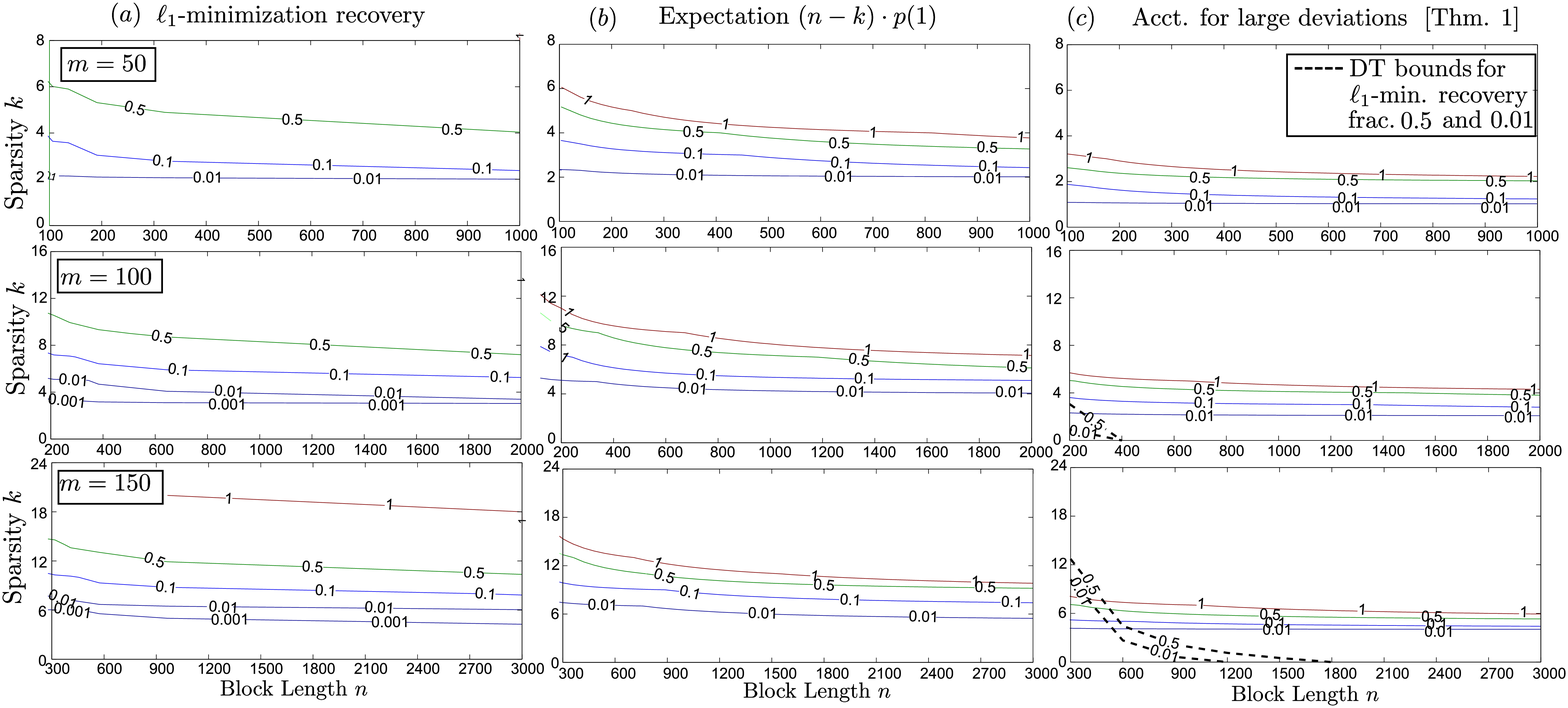},width=.9\linewidth}
		\caption{Gaussian case. Comparing $(a)$ empirical results for $\ell_1$-minimization recovery, $(b)$ mean parameter $(n-k)\cdot p(1)$ (empirically obtained), and $(c)$ after accounting for large deviations (Thm. \ref{thm:Ustat}). We show cases $m=50, 100$ and $150$. We also compare with Donoho \& Tanners' (DT) large deviation bounds~\cite{DTExp}.}
		\label{fig:CompL1}
		\vspace*{-15pt}
\end{figure*}
%


The above experiments suggest the deviation error $\eps_n(a)$ in Theorem \ref{thm:Ustat} to be over-conservative. 
Fortunately in the next two subsections (pertaining to U-statistics treastise of $\ell_1$-recovery Theorem \BargThm~(Section \ref{ssect:ell1}), and LASSO recovery Theorem \CandThm~(Subsection \ref{ssect:Lasso})), this conservative-ness does not show up from a rate standpoint (it only shows up in implicit constants).
In fact by empirically ``adjusting'' these constants, we find good measurement rate predictions (akin to moving from Figure \ref{fig:CompL1}$(c)$ to $(b)$).

\subsection{Rate analysis for $\ell_1$-recovery (Theorem \BargThm)} \label{ssect:ell1}

In ``worst-case'' analysis, it is well-known that it is sufficient to have measurements $m$ on the order of $k\log(n/k)$, in order to have the restricted isometry constants $\delta_k$ defined by (\ref{eqn:RIC}), satisfy the conditions in Theorem \RIPthm.
We now go on to show that for ``average-case'', a similar expression for this rate can be obtained. To this end we require tail bounds on salient quantities. Such bounds have been obtained for the \emph{small projections} condition, see~\cite{Cand2008,Barga,Eldar2009}, where typically an equiprobable distribution is assumed over the sign-vectors $\matt{\beta}_\jp$. 
To our knowledge these techniques were born from considering deterministic matrices. 
Since $\Sens$ is randomly sampled here, we proceed slightly differently (though essentially using similar ideas) without requiring this random signal model.
For simplicity, the bound assumes zero mean matrix entries, either i) Gaussian or ii) bounded.


\begin{pro} \label{pro:proj}
Let $\pmb{A}$ be an $m \times n$ random matrix, whereby its columns $\A_i$ are identically distributed. 
Assume every entry $A_{ij}$ of $\A$ has zero mean, i.e., $\E A_{ij} = 0$.
Let every $A_{ij}$ be either i) Gaussian with variance $1/m$, or ii) bounded RVs satisfying $|A_{ij}| \leq 1/\sqrt{m}$.
Let the rows $[A_{i1},A_{i2},\cdots, A_{in}]$ of $\A$ be IID.

Let $\S$ be a size-$k$ subset, and let index $\omega$ be outside of $\S$, i.e., $\omega\notin\S$. Then for any sign vector $\s$ in $\{-1,1\}^k$, we have 
\ifthenelse{\boolean{dcol}}{
\begin{align}
\Pr\left\{ \left|(\pmb{A}_\S^\dagger \pmb{A}_\omega)^T \s \right| > a \right\} 
\leq &~2\exp\left( - \frac{m a^2 \delta}{2 k} \right) \nn
&+ \Pr\{ \eigm(\A_\S) \leq \delta\} \label{eqn:proj0}
\end{align}}{ 
\bea
\Pr\left\{ \left|(\pmb{A}_\S^\dagger \pmb{A}_\omega)^T \s \right| > a \right\} 
\leq  2\exp\left( - \frac{m a^2 \delta}{2 k} \right)
+ \Pr\{ \eigm(\A_\S) \leq \delta\} \label{eqn:proj0}
\eea}
for any positive $\delta \in \Real$.
\end{pro}
{
\renewcommand{\A}{\pmb{A}}
\renewcommand{\b}{\pmb{B}}
\newcommand{\w}{\omega}
\newcommand{\Ev}{\mathcal{E}}
\renewcommand{\eigM}{\varsigma_{\scriptsize \mbox{\upshape max}}}
\begin{proof}
For $\tau \in \Real$, let $\Ev(\tau) = \{\s^T (\A_\S^T \A_\S)^\dagger \s \leq \tau \}$ where $\Ev(\tau)$ is an probabilistic event. Let $\Ev_c(\tau)$ denote the complementary event.  
Bound the probability as
\ifthenelse{\boolean{dcol}}{
\begin{align}
 \Pr\left\{ \left|(\pmb{A}_\S^\dagger \pmb{A}_\omega)^T \s \right| > a \right\} 
 \leq& \Pr\left\{\left. \left|(\pmb{A}_\S^\dagger \pmb{A}_\omega)^T \s \right| > a  \right| \Ev(\tau) \right\} \nn
 &+  \Pr\{\Ev_c(\tau)\}.
  \label{eqn:proj1}
\end{align}
}{ 
\bea
 \Pr\left\{ \left|(\pmb{A}_\S^\dagger \pmb{A}_\omega)^T \s \right| > a \right\} \leq 
 \Pr\left\{\left. \left|(\pmb{A}_\S^\dagger \pmb{A}_\omega)^T \s \right| > a  \right| \Ev(\tau) \right\} +  \Pr\{\Ev_c(\tau)\}.
  \label{eqn:proj1}
\eea}

We upper bound the first term as follows.
Denote constants $c_1,c_2, \cdots, c_m$. 
For entries $(\A_\w)_i$ of $\A_\w$, consider the sum $ \sum_{i=1}^m c_i \cdot ( m^{-\half}\A_\w)_i = \frac{1}{m}\sum_{i=1}^m c_i X_i$, where RVs $X_i$ satisfy $X_i = (\sqrt{m}  \A_\w)_i$. 
By standard arguments (see Supplementary Material \ref{sup:stdbound}) we have the double-sided bound  $\Pr\left\{\left|\sum_{i=1}^m c_i X_i \right| >  m t \right\} \leq 2\exp\left(-(mt)^2/( 2 \cdot \Norm{\mat{c}}{2}^2) \right)$, where vector $\mat{c}$ equals $[c_1,c_2,\cdots, c_m]^T$.


Next write $(\A_\S^\dagger \A_\w)^T \s = ( \sqrt{m} \cdot \s^T \A_\S^\dagger ) (m^{-\half} \A_\w)$. 
When conditioning on $ \s^T \A_\S^\dagger $, then $ \sqrt{m} \cdot \s^T \A_\S^\dagger $ is fixed, say equals some vector $\mat{c}$. Put $X_i = (\sqrt{m}  \A_\w)_i$ and $X_i$'s are independent (by assumed independence of the rows of $\pmb{A}$). 
Then use the above bound for $\Pr\left\{\sum_{i=1}^m c_i X_i >   t \right\}$, set $t=a$ and conclude
\ifthenelse{\boolean{dcol}}{
\begin{align}
	&\Pr\left\{\left. \left|(\pmb{A}_\S^\dagger \pmb{A}_\omega)^T \s \right| > a  \right| \s^T \A_\S^\dagger \right\} \nn
  &\leq 2\exp\left(- \frac{(m a)^2}{ 2 m \Norm{\s^T \A_\S^\dagger}{2}^2 } \right) = 2\exp\left(- \frac{m a^2}{ 2  \cdot \s^T (\A_\S^T \A_\S)^\dagger \s  } \right),
\end{align}}{
\[
	\Pr\left\{\left. \left|(\pmb{A}_\S^\dagger \pmb{A}_\omega)^T \s \right| > a  \right| \s^T \A_\S^\dagger \right\} 
  \leq 2\exp\left(- \frac{(m a)^2}{ 2 m \Norm{\s^T \A_\S^\dagger}{2}^2 } \right)
  = 2\exp\left(- \frac{m a^2}{ 2  \cdot \s^T (\A_\S^T \A_\S)^\dagger \s  } \right),
\]}
where the last equality follows from the identity $\A_\S^\dagger (\A_\S^\dagger)^T = (\A_\S^T \A_\S)^\dagger $.
Further conclude that the first term in (\ref{eqn:proj1}) is bounded by $2\exp( - m a^2/(2 \tau))$, due to further conditioning on the event $\Ev(\tau) = \{\s^T (\A_\S^T \A_\S)^\dagger \s \leq \tau \}$.

To bound the second term, let $\eigM(\mat{A})$ denote the maximum eigenvalue of matrix $\mat{A}$.
Since $\A_\S^T \A_\S$ is positive semidefinite, note that $\s^T (\A_\S^T \A_\S)^\dagger \s  $ is upper bounded by $\Norm{\s}{2}^2 \cdot \eigM((\A_\S^T \A_\S)^\dagger) $, which equals $  k \cdot \eigM((\A_\S^T \A_\S)^\dagger) $. 
Furthermore $\eigM((\A_\S^T \A_\S)^\dagger) \leq 1/\eigm(\A_\S)$, where here $\sigm(\mat{A})$ is the minimum singular value of $\mat{A}$. 
Thus $\Pr\{\Ev_c(\tau) \} \leq \Pr\{k/\eigm(\A_\S) > \tau\}$.
Finally put $\tau = \delta k$ 
to get $\Pr\{\Ev_c(\tau) \} \leq \Pr\{ \eigm(\A_\S) \leq \delta\Iv\}$.
\end{proof}
}


\renewcommand{\fn}{\footnote{For $m \geq 2k$, we have $\Pr\{\sigm(\pmb{A}) < c \cdot 0.29 - t \} \leq \Pr\{\sigm(\pmb{A}) < 1 - c \cdot  \sqrt{k/m} - t \}\leq \exp(-m t^2/c_1)$ for some constants $c,c_1$, where $\pmb{A}$ has size $m\times k$ and with proper column normalization. For simplicity we drop the constant $c$ in this paper; one simply needs to add $c$ in appropriate places in the exposition. In particular for the Gaussian and Bernoulli cases $c=1$, and $c_1=2$ and $c_1=16$, respectively, see Theorem B,~\cite{Lim2}.}}

Proposition \ref{pro:proj} is used as follows. First recall that previous Proposition \ref{cor:count} allows us to upper bound the fraction $u_2$ of sign-subset pairs $(\s, \S)$ \emph{failing} the \emph{small projections} condition, with the (scaled) U-statistic $(n-k)\cdot U_n(a_2)$ with kernel $g$ in (\ref{eqn:g3}) and $|\S|=k$.
By Theorem \ref{thm:Ustat} the quantity $(n-k)\cdot U_n(a_2)$ concentrates around $(n-k)\cdot p(a_2)$, where $p(a_2)=\E g(\A_\R,a_2)$, where $g$ in (\ref{eqn:g3}) is defined for size-$(k+1)$ subsets $\R$.
We use Proposition \ref{pro:proj} to upper estimate $p(a_2)$ using the RHS of (\ref{eqn:proj0}).
Indeed verify that $p(a_2) = 2^{-k} \sum_{\jp} \Pr\{ |(\pmb{A}_\S^\dagger \pmb{A}_\omega)^T \s_\jp | > a_2 \}$ for any $S$ and $\omega \notin S$, and the bound (\ref{eqn:proj0}) holds for any $\s=\s_\ell$. Now $p(a_2)$ is bounded by two terms. 
By $u_2 \leq (n-k) \cdot U_n(a_2)$, thus to have $u_2$ small, we should have the (scaled) first term $2(n-k)\cdot\exp(-ma_2^2\RIC/(2 k))$ of (\ref{eqn:proj0}) to be at most some small fraction $u$. This requires 
\bea
 m \geq \const \cdot k \log \left(\frac{n-k}{u} \right) \label{eqn:rate}
\eea
with $\const = 2/(a_2^2\RIC)$ (and we dropped an insignificant $\log 2$ term). 
Next, for $m \geq 2k$ and $\RIC < (0.29)^2$, we can bound\fn~the second term $\Pr\{ \eigm(\A_\S) \leq \delta\}$ of (\ref{eqn:proj0}) by $\exp(-m\cdot (0.29-\sqrt{\RIC})^2/c_1)$ where $c_1$ is some constant, see~\cite{Rudel}, Theorem 5.39. 
Roughly speaking, $\eigm(\A_\S) \geq 0.29$ with ``high probability''.
Figures \ref{fig:Gauss} and \ref{fig:BernUnif} (in the previous Section \ref{sect:Dist}) empirically support this fact.
Again to have $u_2$ small the second term of (\ref{eqn:proj0}) must be small. This requires $ (n-k) \cdot \exp(-m \cdot (0.29-\sqrt{\RIC})^2/c_1) \leq  u $ for some small fraction $u$, in which it suffices to have $m$ satisfy (\ref{eqn:rate})
with $\const = c_1/(0.29-\sqrt{\RIC})^2$. 

\newcommand{\w}{\omega}
For the \emph{invertability} condition in Theorem \BargThm, we also need to upper bound the corresponding fraction $u_1$ of size-$k$ subsets $\S$. We simply use an U-statistic $U_n(a_1)$ with kernel $g(\mat{A},a_1) = \Ind{\sigm(\mat{A}) > a_1}$ for some positive $a_1$ (see also Theorem \CandThm).
Here Proposition \ref{cor:count} is not needed.
To make $p(a_1)$ small, where $p(a_1) = \E g(\pmb{A}_\S,a_1)$, use the previous bound $p(a_1) \leq \exp(-m \cdot (0.29-a_1)^2/c_1) $, where we set $a_1 = \sqrt{\RIC}$ with $a_1 \leq 0.29$.
Clearly $p(a_1)$ cannot exceed some fraction $u$, if $m$ satisfies (\ref{eqn:rate}) with $\const = c_1/(0.29-a_1)^2$.

For the time being consider \emph{exactly} $k$-sparse signals $\matt{\x}$. 
In this special case the \emph{worst-case projections} condition in Theorem \BargThm~is superfluous (\textit{i.e.}, with no consequence $a_3$ can be arbitrarily big) -  
only \emph{invertability} and \emph{small projections} conditions are needed.
While we have yet to consider the large deviation error $\eps_n(a)$ from Theorem \ref{thm:Ustat}, doing so will not drastically change the rate. For $U_n(a)$ with kernel $g$ and $p(a)$, where $p(a)=\E g(\pmb{A},a)$, almost surely
\ifthenelse{\boolean{dcol}}{
\bea
\!\!\!\!\!\!\!\!\!U_n(a) \leq p(a) + \eps_n(a) 
\!\!\!&\leq& \!\!\!(p(a))^\half + \sqrt{2 p(a) \w^{-1} \log \w } \nn
\!\!\! &\leq& \!\!\! (p(a))^\half \left(1 +  \sqrt{2 \w^{-1} \log \w }\right) \label{eqn:Ubound}
\eea}{
\bea
U_n(a) \leq p(a) + \eps_n(a) \leq (p(a))^\half + \sqrt{2 p(a) \w^{-1} \log \w } \leq (p(a))^\half \left(1 +  \sqrt{2 \w^{-1} \log \w }\right) \label{eqn:Ubound}
\eea}
where the second inequality follows because $p(a) \leq 1$, and by setting $\w = n/k$.
Taking $\log$ of the RHS, we obtain $(1/2) \log p(a) + \log (1 + \sqrt{2 \w^{-1} \log \w }) $.
Note $\log (1 + \sqrt{2 \w^{-1} \log \w }) \leq \sqrt{2 \w^{-1} \log \w }$, since $\log(1+\alpha) \leq \alpha$ holds for all positive $\alpha$.
For the \emph{small projections} condition, bound $(p(a))^\half$ by the sum of the square-roots of each term in (\ref{eqn:proj0}). 
Then to have $u_2 \leq (n-k) \cdot U_n(a_2) \leq 2 u$, it follows similarly as before that it suffices that (see Supplementary Material \ref{sup:rate})
\bea
 m \geq \const \cdot k \left[ \log \left(\frac{n-k}{u} \right) + \sqrt{2 \cdot (k/n) \log (n/k) } \right] \label{eqn:rate2}
\eea
with $\const =  \max(4/(a_2^2\RIC), 2c_1/(0.29-\sqrt{\RIC})^2)$ where we had set $\sqrt{\RIC} = a_1$ (we dropped an insignificant $\log 2$ term). 
For \emph{invertability} condition do the same. 
To have $u_1 = U_n(a_1) \leq u$ it suffices that $m$ satisfies (\ref{eqn:rate2}) with the same $\const$.
Observe that the term $\sqrt{2 \cdot (k/n) \log (n/k) }$ is at most 1, and vanishes with high undersampling (small $k/n$).
Hence (\ref{eqn:rate}) and (\ref{eqn:rate2}) are similar from a rate standpoint.



\renewcommand{\fn}{\footnote{Comparing (\ref{eqn:rate2}) and (\ref{eqn:rate}) and the respective expressions for $\const$, dropping $\const$ from 4 to 1.8 is akin to ignoring the deviation error $\eps_n(a)$. 
This, and as Figure \ref{fig:CompL1} suggests, the U-statistic ``means'' $(n-k) \cdot p(1)$ seem to predict recovery remarkably well, with similar rates to (\ref{eqn:rate2}), and inherent $\const$ smaller than that derived here.}}


We conclude the following: for exactly $k$-sparse signals
the rate (\ref{eqn:rate2}) suffices to recover at least $1-3u$ fraction of sign-subset $(\s,\S)$ pairs. While $\const$ in (\ref{eqn:rate2}) must be at least $4$ (recall that Figure \ref{fig:CompL1}$(c)$ was somewhat pessimistic), 
for matrices with Gaussian entries 
we empirically find that $\const$ is inherently smaller, whereby $\const \approx 1.8$.
This is illustrated in Figure \ref{fig:MeasRates}, for two fractions $0.1$ and $0.01$ of unrecoverable sign-subset pairs.
We observe good match with simulation results shown in the previous Figure \ref{fig:CompL1}$(a)$, and quantities\fn~$(n-k) \cdot p(1)$ plotted in Figure \ref{fig:CompL1}$(b)$.
For example, $m=150$ suffices for a 0.01 fractional recovery failure, for $n=300 \sim 1000$ and $k = 6\sim 7$, and for 0.1 fraction then $k = 7 \sim 10$.
We conjecture possible improvment for $\const$.

\ifthenelse{\boolean{dcol}}{ 
\begin{figure}[!t]
	\centering
	  \epsfig{file={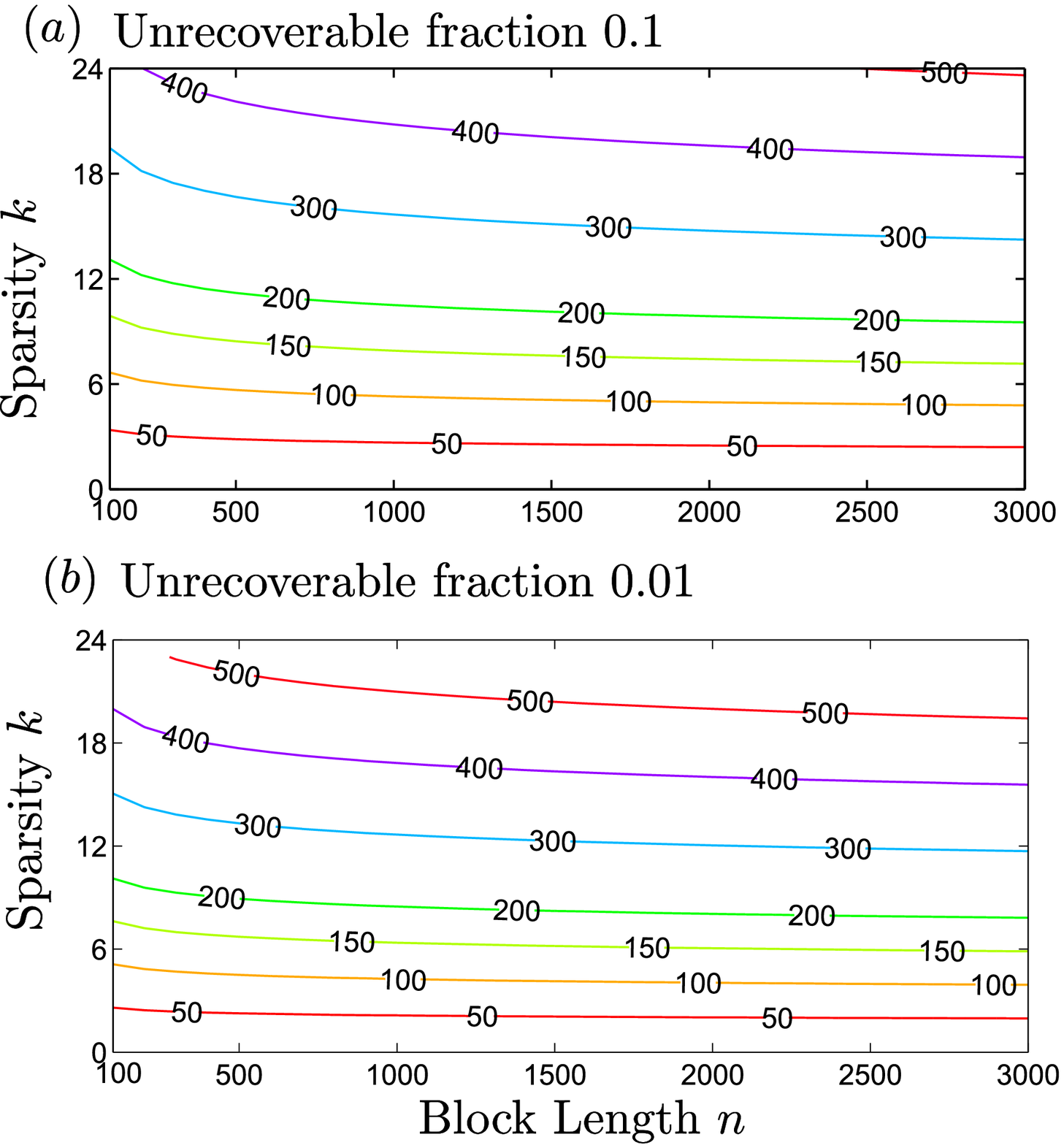},width=.7\linewidth}
		\caption{Measurement rates predicted by equation (\ref{eqn:rate2}), with $\const$ taken to equal $1.8$, required to recover at least $1- 3u = 0.9$ and $0.99$ fractions of sign-subset pairs $(\s,\S)$ (when the signal is exactly $k$-sparse), shown respectively in $(a)$ and $(b)$.}
		\label{fig:MeasRates}
		\vspace*{-10pt}
\end{figure}
}{
\begin{figure}[!t]
	\centering
	  \epsfig{file=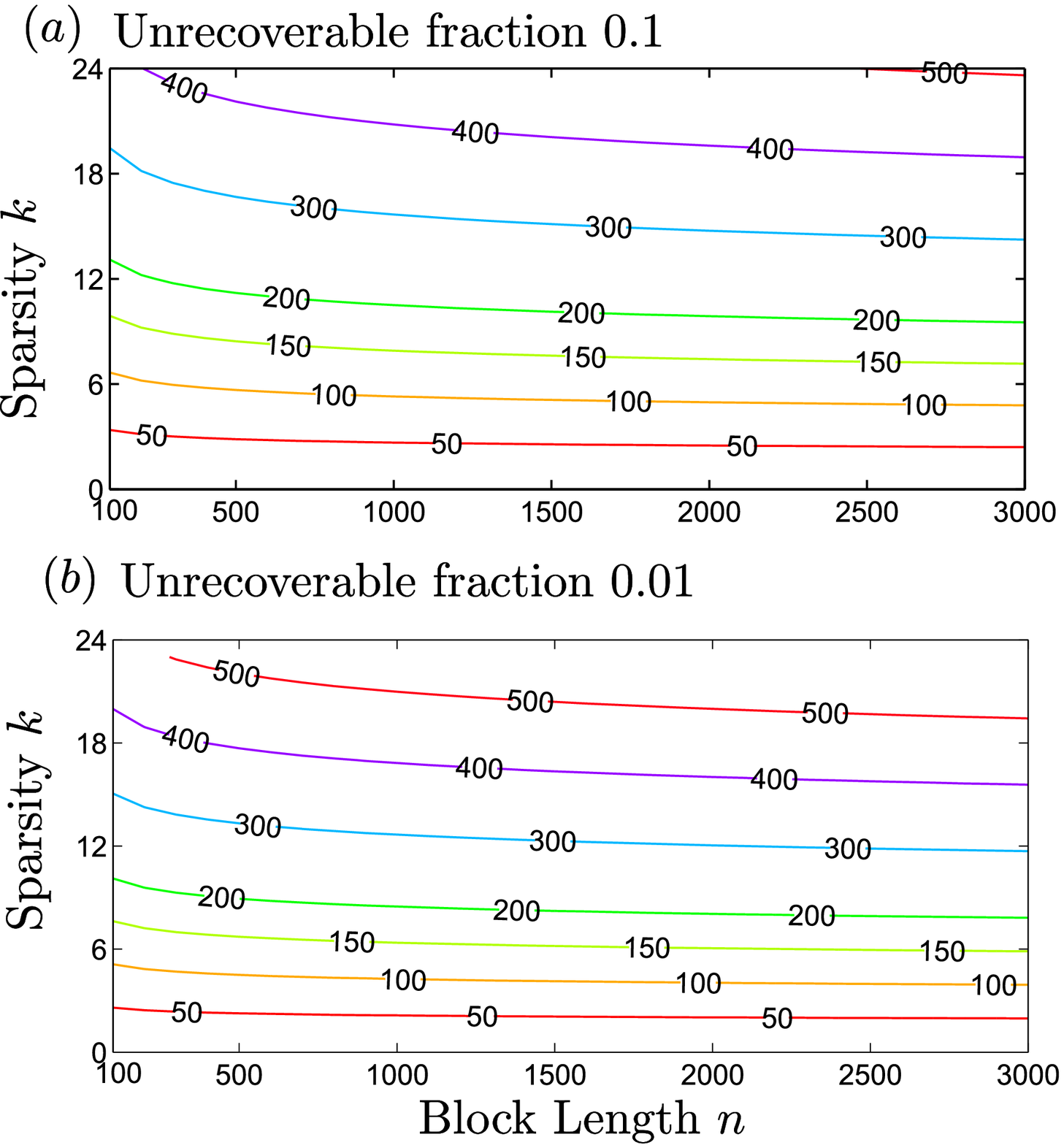},width=.4\linewidth}
		\caption{Measurement rates predicted by equation (\ref{eqn:rate2}), with $\const$ taken to equal $1.8$, required to recover at least $1- 3u = 0.9$ and $0.99$ fractions of sign-subset pairs $(\s,\S)$ (when the signal is $k$-sparse), shown respectively in $(a)$ and $(b)$.}
		\label{fig:MeasRates}
		\vspace*{-10pt}
\end{figure}
}

\renewcommand{\fn}{\footnote{We used an assumption that $(n-k)/u$ is suitably larger than $2$, see Supplementary Material \ref{sup:rate}.}}
In the more general setting for \emph{approximately} $k$-sparse signals, we can also have rate (\ref{eqn:rate2}).
To see this, observe that Proposition \ref{pro:proj} also delivers an exponential bound for the \emph{worst-case projections} condition, see (\ref{eqn:g2}).
This is because $\Norm{\pmb{A}_\S^\dagger \pmb{A}_\omega }{1} = \max_{\jp : \; 1 \leq \jp \leq 2^k} |(\pmb{A}_\S^\dagger \pmb{A}_\omega)^T \s_\jp |$, and we take a 
union bound over $2^k$ terms.
Set $a_3 = a_2 \sqrt{k} $, where $a_2$ and $a_3$ respectively correspond to \emph{small projections} and \emph{invertability} conditions.
Then we proceed similarly as before (see Supplementary Material \ref{sup:rate}) to show\fn~that 
the rate for recovering at least $1-5u$ fraction of $(\s,\S)$ pairs 
suffices to be (\ref{eqn:rate2}). The following is the main result summarizing the exposition so far.

\begin{thm} \label{thm:L1}
Let $\Sens$ be an $m \times n$ matrix, where assume $n$ sufficiently large for Theorem \ref{thm:Ustat} to hold. Sample $\Sens = \pmb{A}$ whereby the entries $A_{ij}$ are IID, and are Gaussian or bounded (as stated in Proposition \ref{pro:proj}). 
Then all three conditions in $\ell_1$-recovery guarantee Theorem \BargThm~for $(\s,\S)$ with $|\S|=k$, with the invertability condition taken as $\sigm(\SensS) \geq a_1$ with $a_1 \leq 0.29$. and with $a_3 = a_1 \sqrt{k}$, are satisfied for $u_1 + u_2 + u_3 = 5u$ for some small fraction $u$, if $m$ is on the order of (\ref{eqn:rate2}) with 
$\const =  \max(4/(a_1a_2)^2, 2c_1/(0.29-a_1)^2)$, and $c_1$ depends on the distribution of $A_{ij}$'s. 
Note $\const \geq 4$.

In the exactly $k$-sparse case where only the first 2 conditions are required, this improves to $u_1 + u_2 = 3u$.
\end{thm}

We end this subsection with two comments on the rate (\ref{eqn:rate2})  derived here for ``average-case'' analysis.
Firstly (\ref{eqn:rate2}) is
very similar to that of $k \log (n/k)$ for ``worst-case'' analysis. 
This justifies the counting employed in previous Subsection \ref{ssect:count}, Proposition \ref{cor:count}, and is reassuring since we know that ``worst-case'' analysis provides the optimal rate~\cite{Candes,Donoho}.
Secondly to have (\ref{eqn:rate2}) hold for the approximately $k$-sparse case, we lose a factor of $\sqrt{k}$ in the error estimate $\Norm{\matt{\x}^*_\S - \matt{\x}_\S}{1}$, as compared to ``worst-case'' Theorem \RIPthm.
This is because we need to set $a_3 = a_2 \sqrt{k} $, as mentioned in the previous paragraph.
However, the ``average-case'' analysis here achieves our primary goal, that is to predict well for system sizes $k,m,n$ when ``worst-case'' analysis becomes too pessimistic.

\renewcommand{\Bas}{\mat{C}}
\newcommand{\am}{\delta_{\scriptsize \mbox{\upshape min}}}
\newcommand{\aM}{\delta_{\scriptsize \mbox{\upshape max}}}
\newcommand{\eeigm}{\varsigma_{\scriptsize \mbox{\upshape min}}}
\newcommand{\eeigM}{\varsigma_{\scriptsize \mbox{\upshape max}}}

\ifthenelse{\boolean{dcol}}{ 
\begin{figure}[!t]
	\centering
	  \epsfig{file=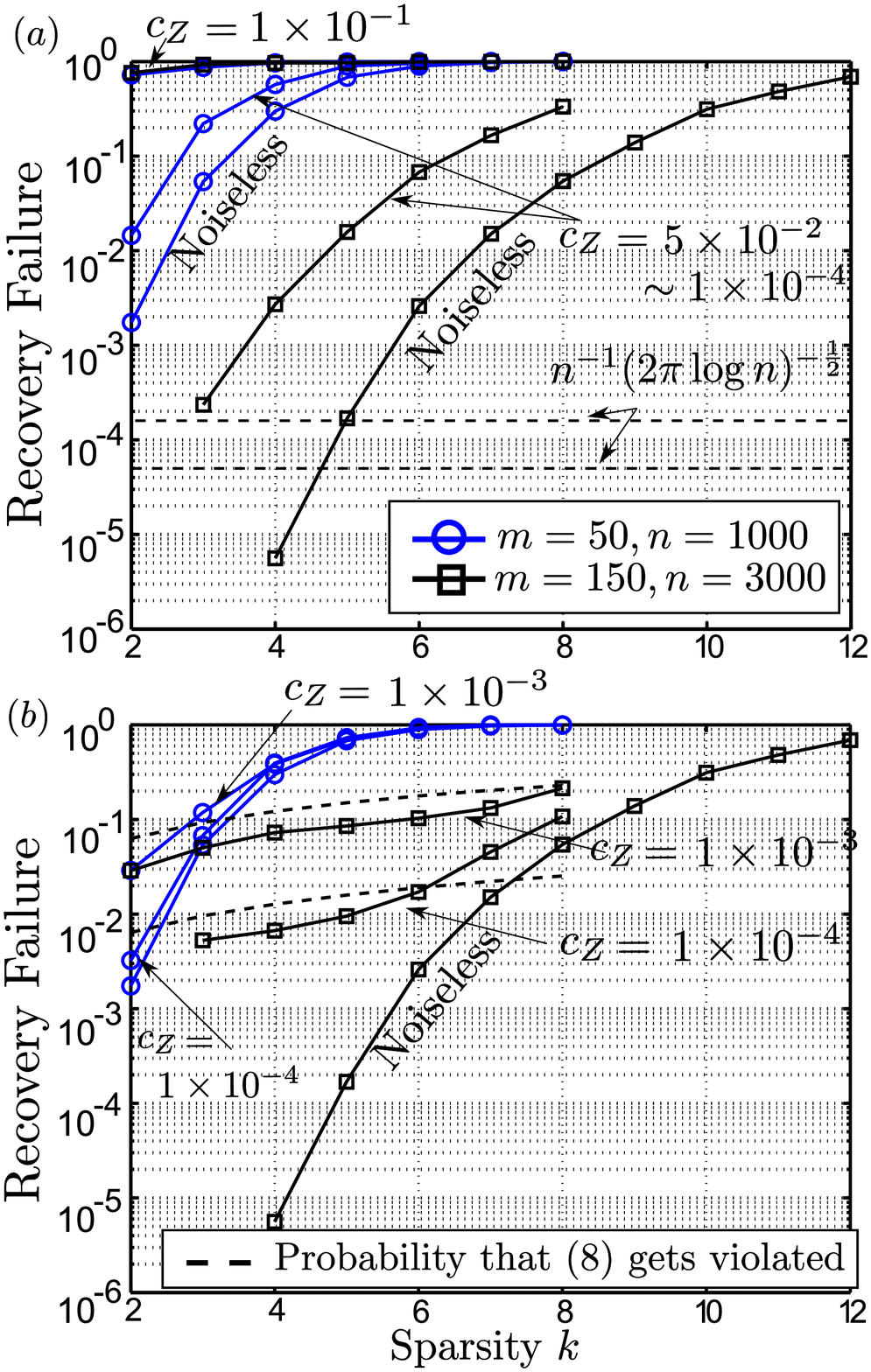,width=.7\linewidth}
		\caption{Empirical LASSO recovery performance, Bernoullli case. In $(a)$ the non-zero signal magnitudes $|\x_i|$ equal 1, and in $(b)$ they are in $\ZO$. Noise variances denoted $\noisesd^2$.}
		\label{fig:Lasso}
		\vspace*{-10pt}
\end{figure}
}{
\begin{figure}[!t]
	\centering
	  \epsfig{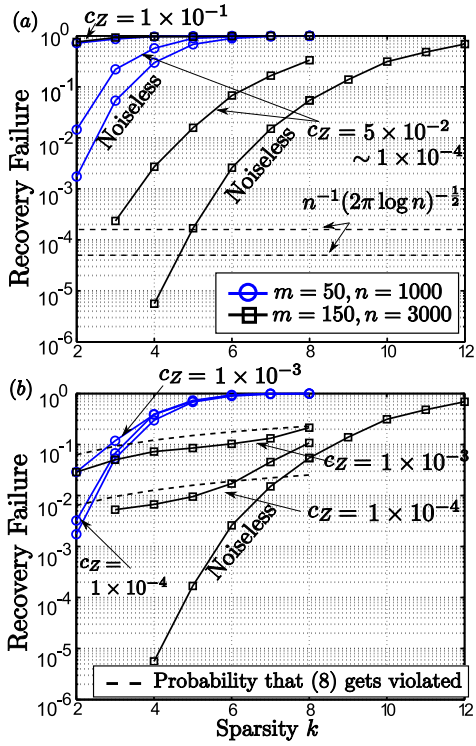}
		\caption{Empirical LASSO recovery performance, Bernoullli case. In $(a)$ the non-zero signal magnitudes $|\x_i|$ equal 1, and in $(b)$ they are in $\ZO$. Noise variances denoted $\noisesd^2$.}
		\label{fig:Lasso}
		\vspace*{-10pt}
\end{figure}
}

\subsection{Rate analysis for LASSO (Theorem \CandThm)} \label{ssect:Lasso}

Next we move on to the LASSO estimate of~\cite{Cand2008}. 
Recall from (\ref{eqn:Lasso}) that the regularizer depends on the noise standard deviation $\noisesd$, and the term $\Reg=(1+a)\sqrt{2 \log n }$ that depends on block length $n$ and some non-negative constant $a$ that we set. 
This constant $a$ impacts performance~\cite{Cand2008}. 
For matrices with Bernoulli entries, Figure \ref{fig:Lasso} shows recovery failure rates for two data sets $m=50, n=1000$ and $m=150, n=1000$; the sparsity patterns (sign-subset pairs $(\s,\S)$) were chosen at random, and failure rates are shown for various sparsity values $k$, and noises $\noisesd$.
In Figure \ref{fig:Lasso}$(a)$ we set $a=0$, and in $(b)$ we set $a=1$.  
Also, in $(a)$ the non-zero signal magnitudes $|\x_i|$ are in $\{1,-1\}$, and in $(b)$ they are in $\ZO$.
The performances are clearly different.
``Threshold-like'' behavior is seen in $(a)$ for both data sets, whereby the performances stay the same for $\noisesd $ in the range $  5 \times 10^{-2} \sim 1 \times 10^{-4}$, and then catastrophically failing for $\noisesd = 1 \times 10^{-1}$. 
However in $(b)$, for various $\noisesd$ the performances seem to be limited by a ``noise-floor''.
We see that in the noiseless limit (more specifically when $\noisesd \rightarrow 0$), the performances become the same. 
In this subsection, we apply U-statistics on the various conditions of Theorem \CandThm, in particular
the \emph{invertability} and \emph{small projections} conditions have already been discussed in the previous subsection.
We account for the observations in Figure \ref{fig:Lasso}.

In the noiseless limit, the previously derived rate (\ref{eqn:rate2}) holds. 
Here, the regularizer in (\ref{eqn:Lasso}) becomes so small that $a$ (equivalently $\Reg$) does not matter. As mentioned in~\cite{Fuchs}, LASSO then becomes equivalent to $\ell_1$-minimization (\ref{eqn:L1}), hence the (noiseless) performances in Figures \ref{fig:Lasso}$(a)$ and $(b)$ are the same.
That is, in this special case the rate (\ref{eqn:rate2}) suffices to recover at least $1-3u$ fraction of $(\s,\S)$.
To test, take $k = 4$, $n=3000$, and fraction $1-3u = 1- 6 \times 10^{-6}$, and with $\const = 1.8$ gives $153$, close to $m$ here which is set to $150$.

In the noisy case, we are additionally concerned with the noise conditions i) and ii), conditions (\ref{eqn:Cand1}) and (\ref{eqn:Cand2}), and \emph{invertability projections}. 
Recall that the noise conditions are satisfied with probability $1- n^{-1} (2\pi \log n)^{-\half}$, that goes to 1 superlinearly~\cite{Cand2008} (Proposition \ref{pro:noise}, Supplementary Material \ref{app:recov}).
The remaining conditions are influenced by the value $a$ set in the $\Reg$ regularization term in (\ref{eqn:Lasso}).

In condition (\ref{eqn:Cand1}), the value $a$ sets the maximal value for $a_2$ (when $a=0$ then $a_2< 0.2929$, and when $a=1$ then $a_2< 0.6464$).
This affects the \emph{small projections} condition, to which constant $a_2$ belongs, which in turn affects performance.
However from a rate standpoint (\ref{eqn:rate2}) still holds, only now 
the value of $\const$ (which has the term $4/(a_2^2\RIC)$) becomes larger. 

In condition (\ref{eqn:Cand2}), the value $a$ affects the size of the term $a_1^{-1} + 2 a_3(1+a)$.
The larger $a$ is, 
the more often (\ref{eqn:Cand2}) fails to satisfy.
Here there are two constants $a_1$ and $a_3$. 
Recall $a_1$ belongs to the \emph{invertability} condition discussed in the previous subsection, which holds with rate (\ref{eqn:rate2}) with $\const =  2c_1/(0.29-a_1)^2$ and $a_1 \leq 0.29$. 
Consider the case where the non-zero signal magnitudes $|\x_i|$ are independently drawn from $\ZO$. Then we observe $(\min_{i \in \S } |\x_i|) < t$ with probability $1 - (1-t)^k$ where $t \in \ZO$ and $|\S|=k$. 
For $t$ set equal to the RHS of (\ref{eqn:Cand2}), this gives the probability that condition (\ref{eqn:Cand2}) fails.
Figure \ref{fig:Lasso}$(b)$ shows good empirical match when setting $a_1 = 0.29$ and $a_3 = 1$, where the dotted curves predict the ``error-floors'' for various $k$, measurements $m=50$ and $m=150$, and noise $\noisesd$.  
In the other case where $|\x_i| =1$ (as in Figure \ref{fig:Lasso}$(a)$), condition (\ref{eqn:Cand2}) remains un-violated as long as $\noisesd$ (and $a_1,a_3,n$) allow the RHS to be smaller than 1. 
Figure \ref{fig:Lasso}$(a)$ suggests that for the appropriate choices for $a_1,a_3$, condition (\ref{eqn:Cand2}) is always un-violated when $\noisesd \leq 5 \times 10^{-2}$, and violated when $\noisesd \geq 1 \times 1^{-1}$.
For more discussion on noise effects see Supplementary Material \ref{sup:Lasso}.


%


\renewcommand{\fn}{\footnote{For $m \geq 2k$ we have $\Pr\{\sigM(\pmb{A}) >   1.71 + t \} \leq \Pr\{\sigM(\pmb{A}) > 1 + \sqrt{k/m} + t \}\leq \exp(-m t^2/c_1)$ for some $c_1$, see~\cite{Rudel}, Theorem 5.39.}}

The constant $a_3$ belongs to the remaining \emph{invertability projections} condition.
The fraction $u_3$ of size-$k$ subsets \emph{failing} the \emph{invertability projections} condition for some $a_3$, can be addressed using U-statistics. 
Consider the bounded kernel $g : \Real^{m\times k} \times \Real \rightarrow \ZO$, set as 
\bea
 g(\mat{A},a) = \frac{1}{2^k}\sum_{\jp=1}^{2^k} \Ind{ (\mat{A}^T\mat{A})^\dagger \s_\jp > a} \label{eqn:g4}
\eea
where $\s_\jp\in \{-1,1\}^k$ and $(\mat{A}^T\mat{A})^\dagger$ is the pseudoinverse of $\mat{A}^T\mat{A}$.
Then $u_3 = U_n(a_3)$, and as before Theorem \ref{thm:Ustat} guarantees the upper bound (\ref{eqn:Ubound}), which depends on
$p(a_3)$ where $p(a_3) = \E g(\pmb{A}_\S,a_3)$.

We go on to discuss a bound on $p(a_3)$ under some general conditions. 
In~\cite{Cand2008}, analysis on $p(a_3)$ (see Lemma 3.5) requires $\eigM(\pmb{A}_\S)  \leq 1.5$, a condition not explicitly required in Theorem \CandThm. 
Also, empirical evidence suggests not to assume that $\eigM(\pmb{A}_\S)  \leq 1.5$.
For $m=150$ and $k=5$ we see from Figure \ref{fig:Lasso} that (in the noiseless limit) the \emph{failure} rate is on the order of $1 \times 10^{-4}$, but in Figure \ref{fig:Gauss}$(b)$ we see $\eigM(\pmb{A}_\S)  > 1.5$ occurs with much larger fraction 0.1. 
Hence we take a different approach.
Using ideas behind Bauer's generalization of \emph{Wielandt's inequality}~\cite{Householder1965}, 
the following proposition allows $\eigM(\pmb{A}_\S)$ to arbitrarily exceed 1.5.
Also, it does not assume any particular distribution on entries of $\pmb{A}$.

{\newcommand{\Ev}{\mathcal{E}}
\begin{pro} \label{pro:invproj}

Let $\S$ be a size-$k$ subset. Assume $k\geq 2$. 
Let $\pmb{A}_\S$ be an $k \times n$ random matrix. 
Let $\am,\aM$ be some positive constants.
For any sign vector $\s$ in $\{-1,1\}^k$, we have
\ifthenelse{\boolean{dcol}}{ 
\begin{align}
&\Pr\left\{ \Norm{(\pmb{A}_\S^T\pmb{A}_\S)^\dagger\s}{\infty}  > \frac{(\sqrt{k} +1 ) \cdot| \tau_k -1 |}{\am^2\cdot(\tau_k + 1)} \cdot  \right\} \nn
&~~~~~~~\leq~  
  \Pr\{ \Ev_c(\am,\aM)   \} \label{eqn:invproj0}
\end{align}}
{\bea
\Pr\left\{ \Norm{(\pmb{A}_\S^T\pmb{A}_\S)^\dagger\s}{\infty}  > \frac{(\sqrt{k} +1 ) \cdot| \tau_k -1 |}{\am^2\cdot(\tau_k + 1)} \cdot  \right\} 
\leq  
  \Pr\{ \Ev_c(\am,\aM)   \} \label{eqn:invproj0}
\eea}
where $\Ev(\am,\aM) = \{ \am \leq \sigm(\pmb{A}_\S) \leq \sigM(\pmb{A}_\S) \leq \aM \}$, and $\Ev_c(\am,\aM)$ is the complementary event of $\Ev(\am,\aM)$, and the constant $\tau_k$ satisfies
\bea
	\tau_k = \tau_k(\aM,\am) =  \left(\frac{\aM }{\am}\right)^2 \cdot  \frac{ 1 + k^{-\half}}{1 -  k^{-\half}} . \label{eqn:invprojtau}
\eea
\end{pro}
}

\renewcommand{\fn}{\footnote{For $m \geq 2k$ we have $\Pr\{\sigM(\pmb{A}) >   1.71 + t \} \leq \Pr\{\sigM(\pmb{A}) > 1 + \sqrt{k/m} + t \}\leq \exp(-m t^2/c_1)$ for some $c_1$, see~\cite{Rudel}, Theorem 5.39.}}

We defer the proof for now.
If $\pmb{A}_\S^T\pmb{A}_\S$ is ``almost'' an identity matrix, then we expect $\Norm{(\pmb{A}_\S^T\pmb{A}_\S)^{-1}\matt{\beta}}{\infty} \approx 1$ for any sign vector $\matt{\beta}$ (hence our above hueristic whereby we set $a_3=1$). 
Proposition \ref{pro:invproj} makes a slightly weaker (but relatively general) statement. 
Now for some appropriately fixed $\aM$ and $\am$, we expect $\Pr\{\mathcal{E}_c(\am,\aM)\}$ in (\ref{eqn:invproj0}) to drop exponentially in $m$.
Just as the term $\Pr\{ \sigm(\A_\S) \leq \am\}$ in Proposition \ref{pro:proj} can be bounded by $\exp(-m\cdot (0.29-\am)^2/c_1)$, 
we can bound\fn~$\Pr\{\sigM(\pmb{A}) >  \aM \} \leq \exp(-m ( \aM - 1.71 )^2/c_1)$ for some $\aM \geq 1.71$.
Roughly speaking, $\sigM(\A_\S) \leq 1.71$ (or $\eigM(\A_\S) \leq 2.92$) with ``high probability''.
We fix $\am = a_1$, where $a_1$ belongs to the \emph{invertability} condition.

So to bound $p(a_3)$, both (\ref{eqn:g4}) and Proposition \ref{pro:invproj} imply $p(a_3) \leq \Pr\{\mathcal{E}_c(\am,\aM)\}$ for $a_3  = (\sqrt{k}+1) \cdot |\tau_k - 1| / (\am^2\cdot (\tau_k+1)) $. 
Now $\Pr\{\mathcal{E}_c(\am,\aM)\} \leq 2 \exp(-m \cdot t^2/c_1) $, where we set $t = \aM - 1.71  = 0.29-a_1 $ and $\am = a_1$.
By (\ref{eqn:Ubound}), the rate (\ref{eqn:rate2}) suffices to ensure $u_3 = U_n(a_3) \leq u$ for some fraction $u$,  
with the same $\const$.
Thus we proved the other main theorem, similar to Theorem \ref{thm:L1}.

\begin{thm} \label{thm:Lasso}
Let $\Sens$ be an $m \times n$ matrix, , where assume $n$ sufficiently large for Theorem \ref{thm:Ustat} to hold. Sample $\Sens = \pmb{A}$ whereby the entries $A_{ij}$ are IID, and are Gaussian or bounded (as stated in Proposition \ref{pro:proj}). 
Then all three \textit{invertability}, \textit{small projections}, and \textit{invertability projections} conditions in LASSO Theorem \CandThm~for $(\s,\S)$ with $|\S|=k \geq 2$, with $a_1 \leq 0.29$, with $a_2$ satisfying (\ref{eqn:Cand1}) for some $a$ set in the regularizer $\Reg$, and with $a_3  = (\sqrt{k}+1) \cdot |\tau_k - 1| / (a_1^2\cdot (\tau_k+1)) $ for $\tau_k=\tau_k(1.42-a_1,a_1)$ in (\ref{eqn:invprojtau}),
are satisfied for $u_1 + u_2 + u_3 = 4u$ for some small fraction $u$, if $m$ is on the order of (\ref{eqn:rate2}) with 
$\const =  \max(4/(a_1a_2)^2, 2c_1/(0.29-a_1)^2)$, and $c_1$ depends on the distribution of $A_{ij}$'s. Note $\const \geq 4$.

In the noiseless limit where only the first 2 conditions are required, this improves to $u_1 + u_2 = 3u$.
\end{thm}

\begin{rem} \label{rem:LassoNoise}
We emphasize again that the rate (\ref{eqn:rate2}) is measured w.r.t. to the three conditions in Theorem \ref{thm:Lasso}. 
The probability for which both noise conditions i) and ii) are satisfied, and for which condition (\ref{eqn:Cand2}) imposed on $\min_{i \in \S}|\x_i|$ is satisfied, require additional consideration. For the former the probability is at least $1 - n^{-1} (2\pi \log n)^{-\half}$, see~\cite{Cand2008}. 
For the latter, it has to be derived based on signal statistics, \textit{e.g.}, for $|\x_i| \in \ZO$ then $(\min_{i \in \S}|\x_i|) > t$ is observed with probability $(1-t)^k$ with $|\S|=k$.
\end{rem}

Note that the choice for $a_3$ in Theorem \ref{thm:Lasso} implies $\Norm{(\pmb{A}_\S^T\pmb{A}_\S)^\dagger\s}{\infty}$ is roughly on the order $\sqrt{k}$. 
Indeed this is true since $\tau_k \geq 1$, and we note $\tau_k = (\aM/\am)^2  + 2k^{-\half}  + o(k^{-\half})$, thus $\tau_k \approx (\aM/\am)^2 $ for moderate $k$. 
Now LASSO recovery also depends on the probability that condition (\ref{eqn:Cand2}) holds. 
Our choice for $a_3$ causes the RHS of (\ref{eqn:Cand2}) to be roughly of the order $\noisesd \sqrt{2 k \log n}$.
Compare this to~\cite{Cand2008} (see Theorem 1.3) where it was assumed that 
$\sigm(\pmb{A}_\S) \leq 1.5$, 
they only require $a_3 = 3$,
\textit{i.e.}, a factor of $\sqrt{k}$ is lost without this assumption (which was previously argued to be fairly restrictive). 
To improve Proposition \ref{pro:invproj}, one might additionally assume some specific distributions on $\pmb{A}$. We leave further improvements to future work.

%
%

{
\newcommand{\Ev}{\mathcal{E}}
\newcommand{\X}{\pmb{X}}
\newcommand{\Xe}{X}
\newcommand{\ang}{\omega}
\renewcommand{\Bas}{\mat{B}}
\newcommand{\D}{\mat{D}}

\renewcommand{\eigM}{\varsigma_{\scriptsize \mbox{\upshape max}}}
\renewcommand{\eigm}{\varsigma_{\scriptsize \mbox{\upshape min}}}
\begin{proof}[Proof of Proposition \ref{pro:invproj}]
For notational convenience, put $\X = (\pmb{A}_\S^T\pmb{A}_\S)^\dagger$. 
Bound the probability
\ifthenelse{\boolean{dcol}}{ 
\begin{align}
 \Pr\left\{ \Norm{\X\s}{\infty}> a\sqrt{k}  \right\} \leq &
 \Pr\left\{\left. \Norm{\X\s}{\infty} > a\sqrt{k}   \right| \Ev(\am,\aM) \right\} \nn&+  \Pr\{\Ev_c(\am,\aM)\}.
  \label{eqn:invproj1}
\end{align}}{
\bea
 \Pr\left\{ \Norm{\X\s}{\infty}> a\sqrt{k}  \right\} \leq 
 \Pr\left\{\left. \Norm{\X\s}{\infty} > a\sqrt{k}   \right| \Ev(\am,\aM) \right\} +  \Pr\{\Ev_c(\am,\aM)\}.
  \label{eqn:invproj1}
\eea}
where we take $a$ to mean
\bea
	a = \frac{|\tau_k-1|}{\tau_k+1} \cdot \frac{1+k^{-\half}}{\sigm^2(\pmb{A}_\S)} \label{eqn:invproja}
\eea
for $\tau_k$ chosen as in (\ref{eqn:invprojtau}).
We claim that every entry $(\X\s)_i$ of $\X\s$ is upper bounded by $a \sqrt{k}$, for $a$ as in (\ref{eqn:invproja}).
Then by definition of $\Ev(\am,\aM)$, the first term in (\ref{eqn:invproj1}) equals $0$ and 
we would have proven the bound (\ref{eqn:invproj0}).


Let $\mat{C}$ denote a $k \times 2$ matrix. 
The first column $\mat{C}$ is be a normalized version of $\s$, more specifically it equals $ k^{-\half}\s_i$. 
The second column equals the canonical basis vector $\mat{c}_i$, where $\mat{c}_i$ is a 0-1 vector whereby $(\mat{c}_i)_j = 1$ if and only if $j=i$.
Consider the $2 \times 2$ matrix $\X'$ that satisfies $\X' = \mat{C}^T \X \mat{C}$. 
This matrix $\X'$ is symmetric (from symmetry of $\X$) and $k^{-\half} (\X\s)_i = X'_{1,2} = X'_{2,1}$ (from our construction of $\mat{C}$). 
That is the entry $X'_{1,2}$ (and $X'_{2,1}$) of $\pmb{X}'$, correspond to the (scaled) quantity $k^{-\half} (\X\s)_i$ that we want to bound.


Condition on the event $\Ev_c(\am,\aM)$, then $\pmb{A}_\S$ has rank $k$ and therefore $\X = (\pmb{A}_\S^T\pmb{A}_\S)^{\dagger} = (\pmb{A}_\S^T\pmb{A}_\S)^{-1}$. 
Let $\det(\cdot)$ and $\Tr(\cdot)$ denote determinant and trace. 
As in~\cite{Householder1965} equation (11), we have
\ifthenelse{\boolean{dcol}}{
\bea
	 \!\!\!\!\!\!\!\!\! 
	 1 - \frac{X'_{1,2}X'_{1,2}}{X'_{1,1} X'_{2,2}} &=& \frac{4 \det (\X') }{(\Tr(\X'))^2 - (X'_{1,1} - X'_{2,2})^2 } \nn
	 &\geq& \frac{4 \eigM(\X') \cdot \eigm(\X') }{(\Tr(\X'))^2} = \frac{4t}{(1+t)^2}
	 \label{eqn:invproj2}
\eea}{
\bea
	 1 - \frac{X'_{1,2}X'_{1,2}}{X'_{1,1} X'_{2,2}} = \frac{4 \det (\X') }{(\Tr(\X'))^2 - (X'_{1,1} - X'_{2,2})^2 } 
	 \geq \frac{4 \eigM(\X') \cdot \eigm(\X') }{(\Tr(\X'))^2} = \frac{4t}{(1+t)^2}
	 \label{eqn:invproj2}
\eea}
where $t = \eigM(\X')/\eigm(\X')$ and $\eigM$ and $\eigm$ respectively denote the maximum and minimum eigenvalues. 
Now $t = \eigM(\X')/\eigm(\X') \geq 1$.
If $t=1$ then $ 4 t/(1+t)^2 = 1$, and for $t \geq 1$ the function $ 4 t/(1+t)^2$ decreases monotonically. 
We claim that $\tau_k$ in (\ref{eqn:invprojtau}) upper bounds $\eigM(\X')/\eigm(\X')$, and (\ref{eqn:invproj2}) then allows us to produce the following upper bound
%
\ifthenelse{\boolean{dcol}}{
\begin{align}
 |X'_{1,2}| &\leq \sqrt{X_{1,1}' X_{2,2}' \cdot \left( 1 - \frac{4 \tau_k}{(1+\tau_k)^2}\right)}  \nn
 &= \sqrt{X_{1,1}' X_{2,2}'} \cdot \frac{|\tau_k-1|}{1 + \tau_k}.
\end{align}}{
\[
 |X'_{1,2}| \leq \sqrt{X_{1,1}' X_{2,2}' \cdot \left( 1 - \frac{4 \tau_k}{(1+\tau_k)^2}\right)} 
 = \sqrt{X_{1,1}' X_{2,2}'} \cdot \frac{|\tau_k-1|}{1 + \tau_k}.
\]}
Bound $(X_{1,1}' X_{2,2}')^\half$ by the maximum eigenvalue $\eigM(\X')$ of $\X'$.
Then, further bound $\eigM(\X')$ by $(1+k^{-\half})/\sigm^2(\pmb{A}_\S)$, which gives the form  (\ref{eqn:invproja}). 
This bound is argued as follows.
For $k\geq 2$, we have the columns in $\mat{C}$ to be \emph{linearly independent}.
Since $\X' = \mat{C}^T \X \mat{C}$ and $\X$ is \emph{positive definite}, it is then clear that $\eigM(\X') \leq \eigM(\mat{C}^T\mat{C}) \cdot \eigM(\X)$.
Now $\mat{C}^T \mat{C}$ is a $2 \times 2$ matrix with diagonal elements 1, and off-diagonal elements $\pm 1/\sqrt{k}$. 
Hence $\eigM(\mat{C}^T \mat{C}) = 1 + k^{-\half}$. 
Also $\eigM(\X) \leq 1/\sigm^2(\pmb{A}_\S)$, and the bound follows.


To finish, we show the claim $\tau_k \geq \eigM(\X')/\eigm(\X')$. By similar arguments as above, it follows that
\[
\frac{\eigM(\X')}{\eigm(\X')} \leq  \frac{\eigM(\mat{C}^T \mat{C})}{\eigm(\mat{C}^T \mat{C})}\cdot \frac{\eigM(\X)}{\eigm(\X)}   = \frac{1 + k^{-\half}}{1 -  k^{-\half}}\cdot \frac{\sigM^2(\pmb{A}_\S)}{\sigm^2(\pmb{A}_\S)} \leq \tau_k
\]
since $\eigm(\X') \geq \eigm(\mat{C}^T\mat{C}) \cdot \eigm(\X)$, and $\eigm(\X')= 1- k^{-\half}$, 
and $\X = (\pmb{A}_\S^T\pmb{A}_\S)^{-1}$. We are done.
%
%
%
%
%
%
%
\end{proof}
}

\section{Conclusion} \label{sect:conc}

We take a first look at U-statistical theory for predicting the ``average-case'' behavior of salient CS matrix parameters.
Leveraging on the generality of this theory, we consider two different recovery algorithms i) $\ell_1$-minimization and ii) LASSO.
The developed analysis is observed to have good potential for predicting CS recovery, and compares well (empirically) with Donoho \& Tanner~\cite{DTExp} recent ``average-case'' analysis for system sizes found in implementations.
Measurement rates that incorporate fractional $u$ failure rates, are derived to be on the order of $ k  [\log((n-k)/u) + \sqrt{2(k/n) \log(n/k)}]$, similar to the known optimal $k\log(n/k)$ rate.
Empirical observations suggest possible improvement for $\const$ (as opposed to typical ``worst-case'' analyses whereby implicit constants are known to be inherently large).

There are multiple directions for future work.
Firstly while restrictive maximum eigenvalue assumptions are avoided (as StRIP-recovery does not require them), the applied techniques could be fine-tuned.
It is desirable to overcome the $\sqrt{k}$ losses observed here for noisy conditions.
Secondly, it is interesting to further leverage the general U-statistical techniques to other different recovery algorithms, to try and obtain their good ``average-case'' analyses.
Finally, one might consider similar U-statistical ``average-case'' analyses for the case where the sampling matrix columns are dependent, which requires appropriate extensions of Theorem \ref{thm:Ustat}.


\section*{Acknowledgment}
The first author is indebted to A. Mazumdar for discussions, and for suggesting to perform the rate analysis.

\appendix 

\subsection{Proof of Theorem \ref{thm:Ustat}} \label{app:proofDev}

\newcommand{\floor}[1]{\left\lfloor#1\right\rfloor }
\newcommand{\perm}{\pi}
\renewcommand{\t}{\eps_n}
\newcommand{\M}{\omega_n}
\newcommand{\Sm}{S_\perm}
For notational simplicity we shall henceforth drop explicit dependence on $a$ from all three quantities $U_n(a), p(a)$ and $\g(\mat{A},a)$ in this appendix subsection. 
While $U_n$ is made explicit in Definition \ref{def:Ustat} as a statistic corresponding to the realization $\Sens = \pmb{A}$, this proof considers $U_n$ consisting of random terms $g(\pmb{A}_\S)$ for purposes of making probabalistic estimates.
Theorem \ref{thm:Ustat} is really a \emph{law of large numbers} result. 
However even when the columns $\pmb{A}_i$ are assumed to be IID, the terms $g(\pmb{A}_\S)$ in $U_n$ depend on each other.
As such, the usual techniques for IID sequences do not apply. 
Aside from large deviation results such as Thm. \ref{thm:Ustat}, there exist \emph{strong law} results, see~\cite{Berk}. The following proof is obtained by combining ideas taken from~\cite{Hoef} and~\cite{Sen}. We use the following new notation just in this subsection of the appendix. Partition the index set $ \{1,2,\cdots,n\}$ into $\M = \floor{n/k}$ subsets denoted $\S_i$ each of size $k$, and a single subset $\R$ of size at most $k$. 
More specifically, let $\S_i = \{(i-1)\cdot k+1,(i-1)\cdot k+2,\cdots, i \cdot k\}$ and let $\R = \{\floor{n/k}\cdot k+1,\floor{n/k}\cdot k+2,\cdots, n  \}$. Let $\perm $ denote a \emph{permutation} (bijective) mapping $ \{1,2,\cdots, n\} \rightarrow \{1,2,\cdots, n\}$. The notation $\perm(\S)$ denotes the set of all \emph{images} of each element in $\S$, under the mapping $\perm$. Following Section 5c in~\cite{Hoef} we express the U-statistic $\V_n$ of $\pmb{A}$ in the form 
\bea
 \V_n  &=& \frac{1}{n!} \sum_{\perm} \left( \frac{1}{\M} \sum_{i=1}^{\M} g(\pmb{A}_{\perm(\S_i)}) \right), \label{eqn:U}
\eea
the first summation taken over all $n!$ possible permutations $\perm$ of $ \{1,2,\cdots,n\}$. To verify, observe that any subset $\S$ is counted exactly $\M \cdot k! (n-k)!$ times in the RHS of (\ref{eqn:U}). 

Recall $p=\E g(\pmb{A}_\S) = \E \V_n$. From the theorem statement let the term $\eps_n^2$ equal $c p (1-p) \cdot \M^{-1} \log \M$ where $c>2$. We show that the probabilities $\Pr\{|\V_n - p| > \t\}$ for each $n$ are small. For brevity, we shall only explicitly treat the upper tail probability $\Pr\{\V_n - p > \t\}$, where standard modifications of the below arguments will address the lower tail probability $\Pr\{-\V_n + p > \t\}$ (see comment in p. 1,~\cite{Hoef}). Using the expression (\ref{eqn:U}) for $\V_n$, write the probability $\Pr\{\V_n - p > \t\}$ for any $h> 0$ as 
\ifthenelse{\boolean{dcol}}{ 
\bea
	\Pr\{\V_n - p > \t\} &\leq&  \E \exp (h( \V_n  - p +   \t)) \nn
	&=& \E \exp\left( \frac{1}{n!} \left( \sum_{\pi}  h(\Sm - p + \t ) \right)\right), \nonumber
\eea}{
\bea
	\Pr\{\V_n - p > \t\} &\leq&  \E \exp (h( \V_n  - p +   \t))
	= \E \exp\left( \frac{1}{n!} \left( \sum_{\pi}  h(\Sm - p + \t ) \right)\right), \nonumber
\eea}
where here $\Sm$ is a RV that equals the inner summation in (\ref{eqn:U}), i.e. $\Sm = \frac{1}{\M}\sum_{i=1}^{\M} g(\pmb{A}_{\perm(\S_i)})$. Using convexity of the function $\exp(\cdot)$ we express 
\bea
\Pr\{\V_n - p > \t\} &\leq&  \frac{1}{n!}  \sum_{\perm} \E \exp(h(\Sm- p + \t)  ). \nonumber
\eea
Now observe that the RV $\Sm$ is an average of $\M$ IID terms $g(\pmb{A}_{\perm(\S_i)})$. This is due to the assumption that the columns $\pmb{A}_i$ of $\pmb{A}$ are IID, and also due to the fact that the sets $\perm(\S_i)$ are disjoint (recall sets $\S_i$ are disjoint).
Hence for any permutation $\perm$, by this independence we have $\E \exp(h\Sm) = (\E \exp( h' \cdot g(\pmb{A}_{\perm(\S_1)})))^{\M} $, where the normalization $h' = h/\M$ bears no consequence. 
The RV $g(\pmb{A}_{\perm(\S_1)}) $ is bounded, i.e. $0 \leq g(\pmb{A}_{\perm(\S_i)}) \leq 1$, and its expectation $\E g(\pmb{A}_{\perm(\S_1)}) $ equals $p$. By convexity of $\exp(\cdot)$ again and for all $h >0$, the inequality 
$e^{h\alpha} \leq e^{h }\alpha + 1 - \alpha  $ 
holds for all $0 \leq \alpha \leq 1$. Therefore putting $\alpha = g(\pmb{A}_{\perm(\S_1)})$ we get the inequality $\exp( h \cdot g(\pmb{A}_{\perm(\S_1)})) \leq 1 + (e^h-1)\cdot g(\pmb{A}_{\perm(\S_1)})  $. 
By the irrelevance of $\perm$ in previous arguments, by putting $\E g(\pmb{A}_{\perm(\S_1)})  = p$
\[
\Pr\{\V_n - p > \t\} \leq e^{-h (\t+p)} \left(1 - p + p e^h\right)^{\M}.
\]
We optimize the bound by putting $pe^h = (1-p)(p+\t)/(1-p-\t)$, see (4.7) in~\cite{Hoef}, to get
\ifthenelse{\boolean{dcol}}{
\begin{align}
&\Pr\{\V_n - p > \t\} \nn 
&\leq \left((1+ \t p\Iv)^{p+\t} (1-\t(1-p)\Iv)^{1-p-\t} \right)^{-\M}. \label{eqn:Uproof1}
\end{align}}{
\bea
\Pr\{\V_n - p > \t\} \leq \left((1+ \t p\Iv)^{p+\t} (1-\t(1-p)\Iv)^{1-p-\t} \right)^{-\M}. \label{eqn:Uproof1}
\eea}
Following (2.20) in~\cite{Sen} we use the relation $\log(1+\alpha) = \alpha - \frac{1}{2}\alpha^2 + o(\alpha^2)$ as $\alpha\rightarrow 0$, to express the logarithmic exponent on the RHS of (\ref{eqn:Uproof1}) as 
\[
 \frac{- \M \t^2\cdot(1+ o(1))}{2p(1-p)}.
\]
Therefore by the form $\t^2 = c p (1-p) \cdot \M^{-1} \log \M$ where $c>2$, for sufficiently large $n$ we have 
\[
	\Pr\{\V_n - p > \t\} \leq \M^{-c/2} < \M^{-1} 
\]
which in turn implies $\sum_{n = k}^\infty  \Pr\{\V_n - p > \t\} < \infty$. Repeating similar arguments for the lower tail probability $\Pr\{-\V_n + p > \t\}$, we eventually prove $\sum_{n = k}^\infty  \Pr\{|\V_n - p| >\t\} < \infty$ which implies the claim.

\newpage
\pagestyle{empty}
\section*{Supplementary Material}
\setcounter{subsection}{0}

\renewcommand{\thefigure}{\Alph{subsection}.\arabic{figure}}


\subsection{Proofs of StRIP-type recovery guarantees appearing in Subsection \ref{ssect:UstatStr}} \label{app:recov}

In this part of the appendix we provide the proofs for the two StRIP-type recovery guarantees discussed in this paper. The following are proofs for Theorems \BargThm~and \CandThm.

{
\newcommand{\err}{\matt{\eps}}
\newcommand{\Sc}{{\S_c}}
\begin{proof}[Proof of Theorem \BargThm,~\textit{c.f.}, Lemma 3,~\cite{Barga}]
Define $\err \in \Real^n$ as $\err = \matt{\x}^* - \matt{\x}$, \emph{i.e.}, $\err$ is the recovery error vector. 
The proof technique closely follows that of Theorem 1.2,~\textit{c.f.},~\cite{RIP}. 
Since $\sgn(\matt{\x}_\S)= \matt{\beta}$, we have the inequality
\bea
	\Norm{(\matt{\x}+\err)_\S}{1} \geq \Norm{\matt{\x}_\S}{1} + \matt{\beta}^T\err_\S. \label{eqn:Arya1}
\eea
Since $\matt{\x}^*$ solves (\ref{eqn:L1}), hence $\Norm{\matt{\x}^*}{1}\leq \Norm{\matt{\x}}{1}$. Putting $\matt{\x}^*=\matt{\x} + \err$, we have
\ifthenelse{\boolean{dcol}}{
\begin{align}
\Norm{\matt{\x}}{1} \geq \Norm{\matt{\x}+\err}{1} &= ~\Norm{(\matt{\x}+\err)_\S}{1} + \Norm{(\matt{\x}+\err)_\Sc}{1} \nn
&\geq ~\Norm{\matt{\x}_\S}{1} + \matt{\beta}^T\err_\S + \Norm{\err_\Sc}{1} - \Norm{\matt{\x}_\Sc}{1},
\end{align}}{
\bea
\Norm{\matt{\x}}{1} \geq \Norm{\matt{\x}+\err}{1} &=& \Norm{(\matt{\x}+\err)_\S}{1} + \Norm{(\matt{\x}+\err)_\Sc}{1} \nn
&\geq& \Norm{\matt{\x}_\S}{1} + \matt{\beta}^T\err_\S + \Norm{\err_\Sc}{1} - \Norm{\matt{\x}_\Sc}{1},
\eea}
where the last step follows the inequality (\ref{eqn:Arya1}), and the triangular inequality.
Re-arranging and putting $\Norm{\matt{\x}_\Sc}{1} = \Norm{\matt{\x}}{1} - \Norm{\matt{\x}_\S}{1}$ 
we get
\bea
\Norm{\err_\Sc}{1} \leq -\matt{\beta}^T\err_\S  + 2\Norm{\matt{\x}_\Sc}{1}. \label{eqn:Arya3}
\eea
We next bound the term $-\matt{\beta}^T\err_\S$ with $|\matt{\beta}^T\err_\S|$, and for now assume that the following claim holds
\bea
|\matt{\beta}^T\err_\S| \leq \Norm{\matt{\beta}^T\SensS^\dagger\Sens_\Sc}{\infty} \cdot \Norm{\err_\Sc}{1}. \label{eqn:Arya2}
\eea
We then proceed to show the bound on $\Norm{\err_\Sc}{1}$ (or $\Norm{\matt{\x}^*_\S - \matt{\x}_\S}{1}$) to complete the first part of the proof. 
Using the \emph{small projections} condition, bound $\Norm{\matt{\beta}^T\SensS^\dagger\Sens_\Sc}{\infty} \leq a_2$ using some $a_2 < 1$. 
This gives a upper bound of $a_2 \cdot \Norm{\err_\Sc}{1}$ on $|\matt{\beta}^T\err_\S|$ in (\ref{eqn:Arya2}). 
Finally use this in (\ref{eqn:Arya3}) get $\Norm{\err_\Sc}{1}  \leq a_2 \Norm{\err_\Sc}{1} + 2\Norm{\matt{\x}_\Sc}{1}$, or equivalently 
$\Norm{\err_\Sc}{1}  \leq  2/(1-a_2) \cdot \Norm{\matt{\x}_\Sc}{1}$.
To show the claim (\ref{eqn:Arya2}), note that $\err$ is in the null-space of $\Sens$, i.e. $\Sens\err=\mat{0}$, or equivalently, $\Sens_\S \err_\S = - \Sens_\Sc \err_\Sc$. Let $\mat{I}$ denote the size-$k$ identity matrix. By the \emph{invertability} condition, the pseudoinverse $\SensS^\dagger$ satisfies $ \SensS^\dagger \SensS = \mat{I}$. 
Hence 
\bea
	\err_\S = - \SensS^\dagger\Sens_\Sc \err_\Sc, \label{eqn:Arya4}
\eea
and take the vector inner product with $\matt{\beta}$ on both sides to obtain $\matt{\beta}^T\err_\S = - \matt{\beta}^T\SensS^\dagger\Sens_\Sc \err_\Sc$. 
Finally (\ref{eqn:Arya2}) holds by taking absolute value of $\matt{\beta}^T\err_\S$, and writing $|\matt{\beta}^T\SensS^\dagger\Sens_\Sc \err_\Sc| \leq \Norm{\matt{\beta}^T\SensS^\dagger\Sens_\Sc}{\infty} \cdot \Norm{\err_\Sc}{1}$.

To second part is to elucidate the bound on $\Norm{\err_\S}{1}$ (or $\Norm{\matt{\x}^*_\Sc - \matt{\x}_\Sc}{1}$).
Starting from the previous relationship (\ref{eqn:Arya4}) we have 
$\Norm{\err_\S}{1} 
= \Norm{\SensS^\dagger\Sens_\Sc \err_\Sc}{1} 
\leq \Norm{\SensS^\dagger\Sens_\Sc}{\infty} \cdot \Norm{\err_\Sc}{1}$.
The result then follows by using the \emph{worst-case projections} condition to bound $\Norm{\SensS^\dagger\Sens_\Sc}{\infty}$ by some positive $a_3 $, and also bounding $\Norm{\err_\Sc}{1}$ using the bound obtained in the first part of this proof. 
\end{proof}
}

For the next two proofs we use the following notation. Let $\mat{I}$ denote the identity matrix, and let $\mat{P}$ denote a projection matrix onto the column subspace of $\Sens_\S$, i.e., $\mat{P} = \Sens_\S \Sens_\S^\dagger$. We first address the proof of Proposition \ref{pro:noise}.

\begin{pro}[\textit{c.f.},~\cite{Cand2008}] \label{pro:noise}
Let $\pmb{Z}$ be a random noise vector, whose components are IID zero mean Gaussian with variance $\noisesd^2$. 
Assume that the matrix $\Sens$ satisfies $\Norm{\col_i}{2} = 1$ for all columns $\col_i$. 
Then the realization $\pmb{Z}= \mat{z}$ satisfies conditions i) and ii) in Theorem \CandThm~with probability at least $1 - n^{-1} (2\pi \log n)^{-\half}$.
\end{pro}


{
\newcommand{\xn}{{\matt{\x}}}
\newcommand{\Sc}{{\S_c}}
\newcommand{\err}{\matt{\eps}}
\newcommand{\yn}{\tilde{\mat{\y}}}
\newcommand{\C}{\mat{c}}
\newcommand{\D}{\mat{d}}
\newcommand{\nsd}{\noisesd}
\newcommand{\Zt}{\tilde{Z}}
\begin{proof}[Proof of Proposition \ref{pro:noise},~\textit{c.f.},~\cite{Cand2008}]
The result will follow by showing i) holds with probability $k\cdot n^{-2} (2 \pi \log n)^{-\half}$, and by showing ii) holds with probability $(n-k)\cdot n^{-2} (2 \pi \log n)^{-\half}$.

For i), first assume each component of $\pmb{Z}$ has variance 1.
Let $\C_i$ denote the $i$-th row of $(\Sens_\S^T \Sens_\S)^{-1}\Sens_\S$, thus we have 
$\Norm{(\Sens_\S^T \Sens_\S)^{-1}\Sens_\S\pmb{Z}}{\infty} = \max_i |\C_i^T \pmb{Z}|$. 
Since $\pmb{Z}$ is Gaussian, thus
\bea
	\Pr\{\Norm{(\Sens_\S^T \Sens_\S)^{-1}\Sens_\S\pmb{Z}}{\infty} > z \} \leq k\cdot \Pr\{ | \Zt| > z\}, \label{eqn:noise1}
\eea
where $\Zt$ is a Gaussian RV with standard deviation at least the $\ell_2$-norm of any row $\C_i$. 
It remains to then upper bound $\Norm{\C_i}{2}$ for all $i$, which follows as
$\Norm{\C_i}{2} \leq \Norm{(\Sens_\S^T \Sens_\S)^{-1}\Sens_\S}{2}$.
The spectral norm $\Norm{(\Sens_\S^T \Sens_\S)^{-1}\Sens_\S}{2}$ is at most the reciprocal of the smallest non-zero singular value of $\Sens_\S$, and by the \emph{invertability} condition for some positive $a_1$, we have 
$\Norm{(\Sens_\S^T \Sens_\S)^{-1}\Sens_\S}{2} \leq a_1^{-1}$.
Then we let $\Zt$ in (\ref{eqn:noise1}) have standard deviation $a_1^{-1}$. 
Equivalently, 
\ifthenelse{\boolean{dcol}}{
\begin{align}
\Pr\{\Norm{(\Sens_\S^T \Sens_\S)^{-1}\Sens_\S^T\pmb{Z}}{\infty} > z \} 
&\leq k\cdot \Pr\{ | Z| > a_1 \cdot z\} \nn
&\leq 2k \cdot f_Z(a_1 z)/(a_1 z)
\end{align}}{
\[
\Pr\{\Norm{(\Sens_\S^T \Sens_\S)^{-1}\Sens_\S^T\pmb{Z}}{\infty} > z \} \leq k\cdot \Pr\{ | Z| > a_1 \cdot z\}
\leq 2k \cdot f_Z(a_1 z)/(a_1 z)
\]}
where $Z$ is a standard normal RV with density function $f_Z(z)$.
Generalizing to the case where each component of $\pmb{Z}$ has variance $\nsd$, the upper bound becomes $2k\cdot f_Z( (a_1 z) /\nsd)/( (a_1 z) /\nsd)$. Put $z = (\nsd \sqrt{2\log n})/a_1$ to get the claimed probabilistic upper estimate $k\cdot n^{-2} (2 \pi \log n)^{-\half}$.

For ii) we proceed similarly.
Observe that for any $i\notin \S$, we have $\Norm{\col_i^T (\mat{I}-\mat{P})}{2} \leq \Norm{\col_i}{2} = 1$. Then put $z = \nsd 2\sqrt{\log n}$ in case ii) to get the claimed probabilistic upper estimate $(n-k)\cdot n^{-2} (2 \pi \log n)^{-\half}$.
\end{proof}
}

{
\newcommand{\xn}{{\matt{\x}}}
\newcommand{\Sc}{{\S_c}}
\newcommand{\err}{\matt{\eps}}
\newcommand{\yn}{\tilde{\mat{\y}}}
\newcommand{\C}{\mat{c}}
\newcommand{\D}{\mat{d}}
\newcommand{\nsd}{\noisesd}
\newcommand{\Zt}{\tilde{Z}}
\newcommand{\maye}{\matt{\x}'}
\begin{proof}[Proof of Theorem 1.3,~\textit{c.f.},~\cite{Cand2008}]
We shall show that any signal $\matt{\x}$ with sign $\matt{\beta}$ and support $\S$, assuming $(\matt{\beta},\S)$ satisfy all three \emph{invertability}, \emph{small projections}, and \emph{invertability projections} conditions together with (\ref{eqn:Cand1}) and (\ref{eqn:Cand2}), will have both sign and support successfully recovered.

The proof follows by constructing a vector $\maye$ from $\matt{\x}$ as follows. 
Let $\err$ denote the error $\err = \maye - \matt{\x}$, and $\maye$ is defined by letting $\err$ satisfy
\bea
	\err_\S &=& (\Sens_\S^T \Sens_\S)^{-1} ( \Sens_\S^T \mat{z} - 2 \nsd\Reg \matt{\beta}), \nn
	\err_\Sc &=& \mat{0}. \label{eqn:lasnoise_1}
\eea
Let us first claim that if (\ref{eqn:Cand2}) holds, then the support of $\maye$ equals that of $\matt{\x}$. 
If this is true, then standard
subgradient arguments, see~\cite{Cand2008,Fuchs2005}, will lead us to conclude that $\maye$ must be the unique Lasso (\ref{eqn:Lasso}) solution (i.e., $\maye = \matt{\x}^*$) if i) it satisfies
\bea
	\col_i^T (\yn - \Sens\maye) &=& 2 \nsd\Reg \cdot \sgn(\alpha'_i), \mbox{   if   } i \in \S,  \nn
	|\col_i^T (\yn - \Sens\maye)| &<& 2 \nsd \Reg , ~~~~~\mbox{   if   } i \notin \S,  \label{eqn:lasnoise_2}
\eea
and ii) the submatrix $\Sens_{\S}$ has full column rank. The condition ii) follows from the \emph{invertability} condition, and the latter half of the proof will verify i). Let us first verify the previous claim that both $\maye$ and $\matt{\x}$ have exact same supports. In fact, we go further to verify that $\maye$ and $\matt{\x}$ also have the same signs. 
First check
\ifthenelse{\boolean{dcol}}{
\begin{align}
	\Norm{\err_\S}{\infty} 
	&\leq~ \Norm{(\Sens_\S^T \Sens_\S)^{-1} \Sens^T_\S \mat{z}}{\infty} 
	+ 2\Reg \nsd \cdot  \Norm{(\Sens_\S^T \Sens_\S )^{-1}\matt{\beta}}{\infty} \nn
	&\leq ~a_1^{-1}  \nsd\cdot \sqrt{2 \log n} + 2 a_3 \nsd \cdot \Reg,
\end{align}}{
\[
	\Norm{\err_\S}{\infty} \leq \Norm{(\Sens_\S^T \Sens_\S)^{-1} \Sens^T_\S \mat{z}}{\infty} 
	+ 2\Reg \nsd \cdot  \Norm{(\Sens_\S^T \Sens_\S )^{-1}\matt{\beta}}{\infty}
	\leq a_1^{-1}  \nsd\cdot \sqrt{2 \log n} + 2 a_3 \nsd \cdot \Reg,
\]}
where the final inequality follows from noise condition i) from Proposition \ref{pro:noise}, 
and the \emph{invertability projections} condition which provides the bound $\Norm{(\Sens_\S^T \Sens_\S )^{-1}\matt{\beta}}{\infty} \leq a_3$ for some positive $a_3$. By assumption (\ref{eqn:Cand2}) and comparing with the above upper estimate for $\Norm{\err_\S}{\infty}$, our claim must hold.

Next we go on to verify $\maye$ satisfies (\ref{eqn:lasnoise_2}). We have
\bea
	\yn - \Sens\maye = \mat{z} - \Sens\err 
	= \mat{z} - (\SensS^\dagger)^T \left(\SensS^T\mat{z} - 2 \nsd \Reg \cdot \matt{\beta} \right) \label{eqn:lasnoise_yn}
\eea
where the last equality follows by first writing $\Sens \err = \Sens \err_\S$, then substituting (\ref{eqn:lasnoise_1}), and putting $\SensS^\dagger = (\SensS^T \SensS)^{-1} \SensS^T$. Now because $\SensS^\dagger$ is a right inverse of $\SensS^T$, by left multiplying the above expression by $\SensS^T$ we conclude 
\[
\SensS^T(\yn - \Sens\maye) = 2 \nsd \Reg\cdot \matt{\beta},
\]
which is equivalent to the first set of equations of (\ref{eqn:lasnoise_1}) as we verified before that $\matt{\beta} = \sgn(\maye_\S)$. 
For the second set of equations, observe from (\ref{eqn:lasnoise_yn}) that
\bea
(\mat{I} - \mat{P}) (\yn - \Sens\xn^*) &=& (\mat{I} - \mat{P}) \mat{z}, \nn
\mat{P} (\yn - \Sens\xn^*) &=& 2 \nsd \Reg \cdot (\SensS^\dagger)^T\matt{\beta}, \nonumber
\eea
where the first equality follows because $(\mat{I} - \mat{P}) (\SensS^\dagger)^T = \mat{0}$, and the second equality follows because $\mat{P} (\SensS^\dagger)^T \SensS^T = \mat{P} \mat{P}^T = \mat{P}^2 = \mat{P}$. 
Using the above two identities, we estimate
\ifthenelse{\boolean{dcol}}{
\begin{align}
& \Norm{\Sens_\Sc^T (\yn - \Sens\xn^*)}{\infty} \nn
&\leq~ \Norm{\Sens_\Sc^T (\mat{I} - \mat{P}) (\yn - \Sens\maye)}{\infty} + 
       \Norm{\Sens_\Sc^T \mat{P} (\yn - \Sens\maye)}{\infty} \nn
&= ~\Norm{\Sens_\Sc^T (\mat{I} - \mat{P}) \mat{z}}{\infty} + 
    2 \nsd \Reg \cdot \Norm{\Sens_\Sc^T(\SensS^\dagger)^T\matt{\beta}}{\infty} \nn
&\leq ~\frac{\nsd \sqrt{2} \Reg}{1+a} + 2 \nsd a_2\cdot \Reg , \label{eqn:lasnoise_3}
\end{align}}{
\bea
\Norm{\Sens_\Sc^T (\yn - \Sens\xn^*)}{\infty} 
&\leq& \Norm{\Sens_\Sc^T (\mat{I} - \mat{P}) (\yn - \Sens\maye)}{\infty} + 
       \Norm{\Sens_\Sc^T \mat{P} (\yn - \Sens\maye)}{\infty} \nn
&=& \Norm{\Sens_\Sc^T (\mat{I} - \mat{P}) \mat{z}}{\infty} + 
    2 \nsd \Reg \cdot \Norm{\Sens_\Sc^T(\SensS^\dagger)^T\matt{\beta}}{\infty} \nn
&\leq& \frac{\nsd \sqrt{2} \Reg}{1+a} + 2 \nsd a_2\cdot \Reg , \label{eqn:lasnoise_3}
\eea}
where the upper estimate $(\nsd \sqrt{2} \Reg)/(1 +a) = \nsd 2\sqrt{\log n} $ follows from noise condition ii) stated in Proposition \ref{pro:noise}, and $ \Norm{\Sens_\Sc^T(\SensS^\dagger)^T\matt{\beta}}{\infty} \leq a_2 $ follows from the \emph{small projections} property.
Finally from assuming (\ref{eqn:Cand1}) 
we have $\sqrt{2}(1+a)^{-1} + 2  a_2 < 2 $, and applying to the last member of (\ref{eqn:lasnoise_3})
proves
$\Norm{\Sens_\Sc^T (\yn - \Sens\maye)}{\infty}  < 2 \nsd \Reg$, which verifies $\maye$ satisfies the second set of equations of (\ref{eqn:lasnoise_1}). Thus we verified $\maye = \matt{\x}^*$ which is what we need to complete the proof.
\end{proof}
}

\subsection{Derivation of standard bounds} \label{sup:stdbound}

{
\renewcommand{\A}{\pmb{A}}
\renewcommand{\b}{\pmb{B}}
\newcommand{\Ev}{\mathcal{E}}
\renewcommand{\eigM}{\varsigma_{\scriptsize \mbox{\upshape max}}}
In the Gaussian case note $\E X_i^2 = 1$ and $\E X_i= 0$. Then $\sum_{i=1}^m c_i X_i $ is also Gaussian with variance $\Norm{\mat{c}}{2}^2$. Hence by Markov's inequality we have the (single-sided) inequality 
$\Pr\left\{\sum_{i=1}^m c_i X_i >   t \right\} \leq \exp(-ht + h^2/\Norm{\mat{c}}{2}^2)$ for any $h> 0$. The claim  for the Gaussian case will follow by setting $h = t \cdot \Norm{\mat{c}}{2}^2 / 2$, and noting that for the other side $\Pr\left\{- (\sum_{i=1}^m c_i X_i) >   t \right\} = \Pr\left\{\sum_{i=1}^m c_i X_i >   t \right\}$. 
For the bounded case, note $|X_i| \leq 1$ and $\E X_i= 0$, and the claim follows from Hoeffding's (2.6) in~\cite{Hoef}.
}

\subsection{Derivation of measurement rates} \label{sup:rate}

For the \emph{small projections} condition, start from $p(a_2)$ being bounded by the RHS of (\ref{eqn:proj0}) where $a=a_2$.
As before bound $\Pr\{ \sigm(\A_\S) \leq a_1\} \leq \exp(-m\cdot (0.29-a_1)^2/c_1)$, where we had set $\sqrt{\delta} = a_1$. 
From the identity $\sqrt{\alpha_1} \leq \sqrt{\alpha_2}  + \sqrt{\alpha_3}$ for positive quantities $\alpha_i$, 
it follows from Theorem \ref{thm:Ustat} and (\ref{eqn:Ubound}) that we will have $u_2 \leq (n-k) \cdot U_n(a_3) \leq 2 u$, if we enforce
\bea
\frac{1}{2} \left[ \log 2 + \log(n-k)  - \frac{m (a_1 a_2)^2 }{2k }\right] + t &\leq& \log u, \nn
\frac{1}{2} \left[   \log(n-k)  - \frac{m (0.29-a_1)^2 }{c_1 }\right] + t &\leq& \log u, \nonumber
\eea
where $t = \sqrt{2(k/n)\log(n/k)}$.
Ignoring the $\log2$ term, and using $\sqrt{n-k} \leq n-k$, it follows that (\ref{eqn:rate2}) enforces the two above conditions.

Similarly for the \emph{invertability} condition, to have $u_1 = U_n(a_1) \leq u$ it follows from Theorem \ref{thm:Ustat} and (\ref{eqn:Ubound}) that we need to enforce to second condition above.

For the \emph{worst-case projections} condition, to have $u_3 \leq (n-k) \cdot U_n(a_3) \leq 2 u$ we need to enforce 
\bea
\frac{1}{2} \left[(k+1)\cdot \log 2 + \log(n-k)  - \frac{m (a_1 a_3)^2 }{2k }\right] + t &\leq& \log u, \nn
\frac{1}{2} \left[   k \log 2 + \log(n-k)  - \frac{m (0.29-a_1)^2 }{c_1 }\right] + t &\leq& \log u. \nonumber
\eea
Taking 
\[
k \log \left( \frac{n-k}{u}\right) \geq (k+1)\cdot \log 2  + \log \left( \frac{n-k}{u}\right), 
\]
justifiable for $(n-k)/u $ suitably larger than $2$, the rate (\ref{eqn:rate2}) generously suffices to ensure these 2 conditions.

\subsection{More on noisy LASSO performance} \label{sup:Lasso}

\setcounter{figure}{0}
\begin{figure*}[!t]
	\centering
	  \epsfig{file={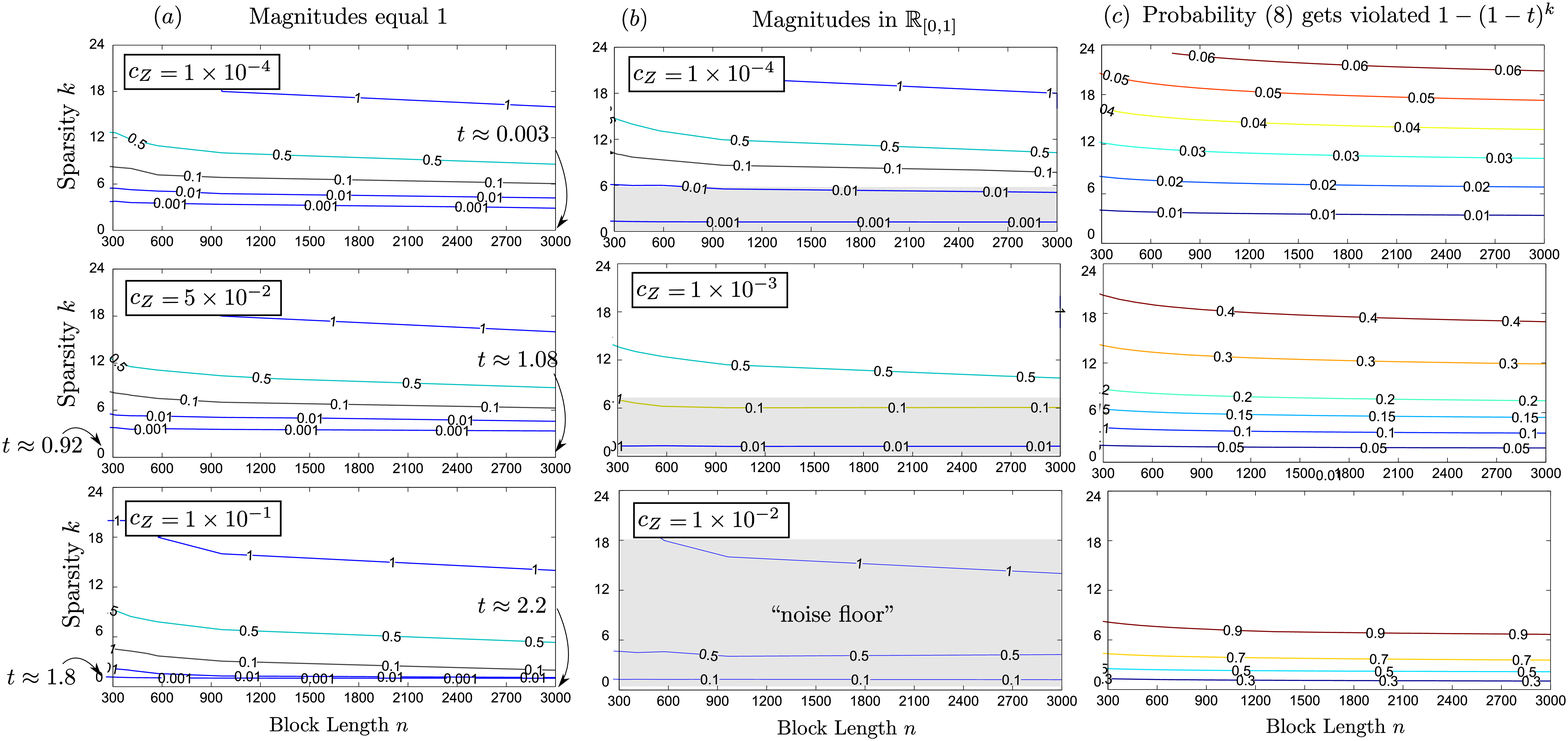},width=.9\linewidth}
		\caption{Empirical LASSO performance shown for $m=150$ for range of $k,n$ values. In $(a)$ the non-zero signal magnitudes $|\x_i|$ equal 1, and in $(b)$ they are in $\ZO$. In $(c)$ we plot a curve (expression) $1-(1-t)^k$ for $t = (3.4 + 2(1+a))\cdot  \noisesd\sqrt{2 \log n}$.}
		\label{fig:Supp}
\end{figure*}

The aim here is to provide more empirical evidence to support observations made in Figure \ref{fig:Lasso} for more block lengths. Here Figure \ref{fig:Supp} shows LASSO performance now for a wider range of $n$. 
We only consider $m=150$, and show various recovery failure rates displayed via contoured lines, for various sparsities $k$ and block lengths $n$.
Figures \ref{fig:Supp}$(a)$ and $(b)$ are companion to Figures \ref{fig:Lasso}$(a)$ and $(b)$, in that they respectively correspond to cases where the non-zero signal magnitudes $|\x_i|$ equal 1 (and $a=0$), and in $\ZO$ (and $a=1$).
That is, for $n=1000$, and $k=4$ and $\noisesd = 1 \times 10^{-4}$, we see the recovery failure is approximately $1\times 10^{-3}$ in both Figure \ref{fig:Supp}$(a)$ and Figure \ref{fig:Lasso}$(a)$.

As mentioned in Subsection \ref{ssect:Lasso} we observe good empirical match when adjusting the term $t = (a_1^{-1} + 2 a_3(1+a)) \cdot \noisesd\sqrt{2\log n}$ (on the RHS of (\ref{eqn:Cand2})) with $a_1 = 0.29$ and $a_3 = 1$. Figure \ref{fig:Supp} provides further support. 
In $(a)$ we show the values of the term $t$ for values $n=300$ and $n=3000$. 
Recall in this case when $t > 1$ condition (\ref{eqn:Cand2}) (and thus recovery) fails.
Observe when $\noisesd = 5 \times 10^{-2}$ the values of $t$ are very close to $1$, and for $\noisesd = 1 \times 10^{-1}$ they exceed $1$. 
This matches with our observation in Figure \ref{fig:Lasso}$(a)$ that $\noisesd = 5 \times 10^{-2}$ is the critical point, beyond which for large $\noisesd$ 
recovery fails catastrophically. 

In $(b)$ and $(c)$ we look at the other case where $|\x_i| \in \ZO$.
Here $(c)$ plots the probability $1-(1-t)^k$ that (\ref{eqn:Cand2}) fails. 
Again the contoured lines delineate a particular fixed value of $1-(1-t)^k$ for various $k,n$ values, whereby we set $t = 7.4 \cdot \noisesd\sqrt{2\log n}$ (recall we used $a=1$ here).
We observe how closely $(c)$ tracks the noise floor regions in $(b)$ (indicated by shading).
More specifically note $t$ really depends on $n$, and
the larger the probabilities $1-(1-t)^k$ get for various $k,n$ in Figure \ref{fig:Supp}$(c)$, this probability overwhelms the LASSO recovery rates in Figure \ref{fig:Supp}$(b)$. 
This matches with our previous observations in Figure \ref{fig:Lasso}$(b)$. 

\end{document}